\documentclass{article}
\usepackage[a4paper, margin=2.5cm]{geometry} 
\usepackage[medium]{titlesec}

\usepackage{amssymb} 
\usepackage{amsmath}
\usepackage{amsthm}

\usepackage{bm}

\usepackage{url}
\usepackage{float}
\usepackage[utf8]{inputenc} 
\usepackage[T1]{fontenc} 
\usepackage{lmodern}
\usepackage{appendix}

\usepackage{bbm}
\usepackage{hyperref}

\newcommand{\R}{{\mathbb R}}
\newcommand{\Z}{{\mathbb Z}}
\newcommand{\Lie}{{\mathcal L}}

\newcommand{\iI}{\mu}
\newcommand{\jI}{\nu}
\newcommand{\kI}{\xi}
\newcommand{\lI}{\chi}

\DeclareMathOperator{\tf}{tf}

\newcommand{\g}{\tilde{g}}
\newcommand{\RS}{\tilde{S}}
\newcommand{\RR}{\tilde{R}}
\newcommand{\tnabla}{\tilde{\nabla}}
\newcommand{\tsquare}{\tilde{\square}}
\newcommand{\tE}{\tilde{E}}
\newcommand{\tB}{\tilde{B}}
\newcommand{\tF}{\tilde{F}}

\newcommand{\tD}{\tilde{D}}
\newcommand{\tG}{\tilde{G}}

\newcommand{\tS}{\tilde{\gamma}}
\newcommand{\tQ}{\tilde{Q}}
\newcommand{\tX}{\tilde{X}}

\newcommand{\tboxdot}{\tilde{\boxdot}}

\newcommand{\mM}{\mathbf{M}}
\newcommand{\mRS}{\mathbf{S}}
\newcommand{\mF}{\mathbf{F}}

\newcommand{\mD}{\mathbf{D}}
\newcommand{\mX}{\mathbf{X}}
\newcommand{\mY}{\mathbf{Y}}
\newcommand{\mT}{\mathbf{T}}
\newcommand{\mnabla}{\boldsymbol{\nabla}}
\newcommand{\msquare}{\bm{\boxdot}}
\newcommand{\mboxplus}{\bm{\boxplus}}
\newcommand{\mg}{\mathbf{g}}
\newcommand{\mf}{\mathbf{f}}
\newcommand{\mr}{\bm{\Omega}}

\newcommand{\rd}{{\mathbbm d}}

\newtheorem{df}{Definition}
\newtheorem{lm}{Lemma}
\newtheorem{thm}[lm]{Theorem}
\newtheorem{prop}[lm]{Proposition}
\newtheorem*{prop*}{Proposition}
\newtheorem{cor}[lm]{Corollary}
\theoremstyle{remark}
\newtheorem{remark}{Remark}

\theoremstyle{remark}
\newtheorem{example}{Example}

\title{Well-posedness of the ambient metric equations and stability of even dimensional asymptotically de Sitter spacetimes}
\author{Wojciech Kami{\'n}ski\\\footnotesize{Institute of Theoretical Physics, Faculty of Physics, University of Warsaw}\\
\footnotesize{Pasteura 5, 02-093 Warsaw, Poland}}

\begin{document}

\maketitle

\begin{abstract}
Vanishing of the Fefferman-Graham obstruction tensor was used by Andersson and Chru{\'s}ciel to show stability of the asymptotically de Sitter spaces in even dimensions. However, existing proofs of hyperbolicity of this equation contain gaps.
We show in this paper that it is indeed a well-posed hyperbolic system with unique up to diffeomorphism and conformal transformations smooth development for smooth Cauchy data.  Our method applies also to equations defined by various versions Graham-Jenne-Mason-Sparling operators. In particular, we use one of these operators to propagate Gover's condition of being almost conformally Einstein. This allows to study initial data also for Cauchy surfaces which cross conformal boundary. As a by-product we show that on globally hyperbolic manifolds one can always choose conformal factor such that Branson Q-curvature vanishes. 
\end{abstract}

\section{Anderson-Fefferman-Graham equation}

An important issue in General Relativity is the long time, asymptotic behaviour of solutions to Einstein's equations. In the case of positive cosmological constant the problem was solved by Friedrich \cite{Friedrich1983}.  He showed that there exists in $4$ dimensions a hyperbolic system of equations for a metric and some derived variables which is satisfied if a metric is conformal to a solution of Einstein equation with a cosmological constant. This allows to study compactified versions of the solutions via conformal Penrose compactification and replace difficult long time analysis by a simpler finite time problem. Future asymptotically simple solutions are those, which conformal compactification extends smoothly to the future boundary Cauchy surface $\overline{\Sigma}_+$. An important example of such a spacetime is de Sitter universe and such spacetimes are often called asymptotically de Sitter. In fact, we need to assume positive cosmological constant in order the conformal boundary surface to be spacelike.  From hyperbolicity of the new system one obtain immediately stability in this class of spacetimes. Moreover, the method gives explicite description of the initial data on the conformal boundary $\overline{\Sigma}_+$. However the Friedrich's method does not extend easily to higher dimension. Other important drawback is that no Lagrangean formulation for it exists. Such formulations are important for analysis of the initial data and conserved charges. 

The alternative method proposed by Anderson in \cite{Anderson} and futher developed by Anderson and Chru\'sciel in \cite{Anderson-Chrusciel} is using the Fefferman-Graham obstruction tensor $H_{\iI\jI}$ which is defined for even dimensions $d\geq 4$ \cite{Fefferman-Graham-1}. We will describe the original definition of \cite{Fefferman-Graham-1} (see \cite{Fefferman-Graham}) in section \ref{sec:ambient}. This tensor can be defined as the variation of the Lagrangean 
\begin{equation}
c_d\int Q(h)\sqrt{h} \rd^dx,
\end{equation}
over the metric $h_{\mu\nu}$, where $Q$ is the Branson curvature \cite{Branson}, a covariant object with nice conformal transformations, and $c_d$ is a constant depending on the dimension. This functional is invariant (up to boundary terms) under both diffeomorphisms and conformal transformations. Consequently,
the $H_{\iI\jI}$ tensor has interesting properties
\begin{enumerate}
\item It is a covariant object built out of the metric and its derivative,
\item It is conformally covariant, namely for $h_{\iI\jI}^1=e^{2\sigma} h^2_{\iI\jI}$, for $\sigma$ a smooth function
\begin{equation}
e^{(d-2)\sigma}H_{\iI\jI}(h^1)=H_{\iI\jI}(h^2).
\end{equation}
\item It is divergenceless $\nabla^{\iI}H_{\iI\jI}=0$ and traceless $H^{\iI}_{\iI}=0$.
\end{enumerate}
Even more remarkably, 
\begin{enumerate}
\item[4] if $h_{\iI\jI}$ satisfies vacuum Einstein equations with a cosmological constant $\Lambda$:  $G_{\iI\jI}(h)=\Lambda h_{\iI\jI}$, then $H_{\iI\jI}=0$ \cite{Fefferman-Graham}.
\end{enumerate}

The method proposed in \cite{Anderson, Anderson-Chrusciel} is to consider conformally invariant Anderson-Fefferman-Graham (AFG) equations 
\begin{equation}\label{eq:AFG}
H_{\iI\jI}=0,
\end{equation}
instead of the Einstein equations and impose the later as a constraint at the initial surface $\overline{\Sigma}_+$. If we want to use this method, it is necessary to prove that the system \eqref{eq:AFG} is well-posed after fixing diffeomorphism and conformal gauge.

However, this is a tricky problem. It is shown in \cite{Anderson-Chrusciel} that one can impose gauge $\square_{\overline{h}}(x^\mu)=0$ and $R_{\overline{h}}=0$ (where $\square_{\overline{h}}$ is a scalar d'Alembert operator for the metric $\overline{h}$, $x^\mu$ are coordinates and $R_{\overline{h}}$ is the Ricci scalar). In this specific gauge the equations take the form
\begin{equation}
\square_{\overline{h}}^{d/2}\overline{h}_{\iI\jI}+F^{\overline{h}}_{\iI\jI}(D^{d-1}\overline{h}_{\iI\jI})=0,
\end{equation}
where $D^m\overline{h}_{\iI\jI}$ denote $m$-th jets of the metric i.e. all derivatives $\partial^k \overline{h}_{\iI\jI}$ for $k\leq m$. The principal symbol is hyperbolic, but the roots have multiplicities, thus the system is not strictly hyperbolic. Such systems are complicated as we can see in the following example:

\begin{example}[Not well-posed]
Consider an equation on $\R\times S^1$ with $x^1$ being time coordinate
\begin{equation}
(\partial_1^2-\partial_2^2)^3\phi+\partial_2(\partial_1+\partial_2)^3\phi=0.
\end{equation}
The principal symbol is hyperbolic (it is a power of d'Alembert operator for a flat metric), but it has multiple characteristics. Functions $\phi_k(x^1,x^2)=e^{i(\omega(k) x^1+kx^2)}$ are solutions for $\omega(k)=\frac{-1-i\sqrt{3}}{2} k^{1/3}$ with property $|\phi_k|=e^{\frac{\sqrt{3}}{2}x^1}$. The smooth initial data on $\Sigma=\{x^1=0\}$
\begin{equation}
\partial_1^n\phi|_{\Sigma}=\sum_{k=0}^\infty i^n\omega(k)^n e^{-k^{1/4}}e^{ikx^2},\quad n=0\ldots 5,
\end{equation}
does not admit a Cauchy development because for every $k\geq 0$ the mode function should behave like $e^{-k^{1/4}}\phi_k(x^1,\cdot)$ for $x^1>0$ but the series $\sum_{k\geq 0} e^{-k^{1/4}}\phi_k(x^1,x^2)e^{ikx^2}$ does not converge even in $L^2(S^1)$.
\end{example}

We see that in order to establish hyperbolicity of the equations with a non-strictly hyperbolic principal part, one needs to control few lower order derivatives of the equation (by so called Levi conditions). Unfortunately, the Fefferman-Graham obstruction tensor is quite complicated and we need to control more and more of these terms the higher dimension we consider. This is the reason, why proofs of well-posedness of the Anderson-Fefferman-Graham equation in \cite{Anderson} and \cite{Anderson-Chrusciel} are not correct. In particular, it is assumed in \cite{Anderson}, that $\square_{\overline{h}}^{d/2}$ is strongly hyperbolic, but our example (for a flat metric) shows that it is not the case. The conditions on the lower order symbols are necessary for the case of multiple characteristics. 
We will provide in this paper a proof for smooth data.
Our method is in fact a modification of approach from \cite{Anderson-Chrusciel}, but the special form of the system needs to be taken into account.

The problem is less complicated in lower dimensions. It is worth to mention that the well-posedness in dimension $4$ was proven in \cite{Guenther}.\footnote{In dimension $4$ the obstruction tensor is proportional to the celebrated Bach tensor and the analysis is less complicated.} Our approach can be regarded as a generalization, which put also \cite{Guenther} in a proper context.

\section{Summary of the results}

The Cauchy problem of the Anderson-Fefferman-Graham (AFG) equation \eqref{eq:AFG} is similar to that of the Einstein equations. The metric itself is a object of equations (see \cite{Choquet-Bruhat}). The initial data on the surface $\Sigma$ will be a set of $d-1$ jets of symmetric tensors fields in $\R\times\Sigma$ at $\Sigma$
\begin{equation}\label{AFG:data}
D^{d-1}h_{\iI\jI}\in C^\infty(\Sigma),
\end{equation}
where $h_{\iI\jI}$ is a Lorentzian metric. We introduce a normal $\vec{N}$ to $\Sigma$ with respect to this metric and we assume that it is a timelike vector.
Assume that \eqref{AFG:data} satisfy constraints (well-defined because we know sufficiently many derivatives)
\begin{equation}\label{AFG:constraints}
H(\vec{N},\cdot)|_\Sigma=0,
\end{equation}
We will consider a specific (local) coordinate system given by conditions introduced in \cite{Anderson-Chrusciel}
\begin{equation}\label{AFG:gauge1}
R=0,\ \forall_\mu\ \square(x^{\iI})=0,
\end{equation}
where $R$ is the Ricci scalar.
As it is shown in \cite{Anderson-Chrusciel} we can always locally transform the metric by diffeomorphism and rescaling, such that these conditions are satisfied. We then show that equations \eqref{eq:AFG} are hyperbolic in this gauge.
The standard analysis \cite{Choquet-Bruhat} thus shows:

\begin{thm}\label{thm:AFG}
The AFG equation \eqref{eq:AFG} with initial data \eqref{AFG:data} and subject to constraints \eqref{AFG:constraints} forms a $C^\infty$ well-posed system. Every two local solutions differ by diffeomorphisms and conformal transformations.
\end{thm}

We will now describe our approach to the problem. Fefferman-Graham obstruction tensor is obtained by ambient metric construction. We consider expansion coefficient of the ambient metric $g^{[k]}_{\mu\nu}$ for $k=0,\ldots \frac{d}{2}-1$	and then the equation of vanishing of the Fefferman-Graham tensor are equivalent to
\begin{equation}\label{eq:FG-1}
S^{[k]}_{\mu\nu}=0,\quad k=0,\ldots, \frac{d}{2}-1,
\end{equation}
where $S^{[k]}_{\mu\nu}$ is  a part of the expansion of the Ricci tensor for the ambient metric. We will briefly describe the ambient construction in Section \ref{sec:ambient}. The AFG equation is obtained by recursive determination of $g^{[k]}_{\mu\nu}$ for $k=1,\ldots \frac{d}{2}-1$ in terms of $g^{[0]}_{\mu\nu}=h_{\mu\nu}$. 

Instead of solving recursively, we will consider these equations as a dynamical system for $g^{[k]}_{\mu\nu}$ for $k=0,\ldots \frac{d}{2}-1$. Every equation in \eqref{eq:FG-1} as coming from Ricci tensor is of second order, but the system is not of hyperbolic type. It is not surprising because we need gauge fixing. Following standard Choquet-Bruhat method we write
\begin{equation}
S^{[k]}_{\mu\nu}=E^{[k]}_{\mu\nu}+\partial_\mu G^{[k]}_\nu+\partial_\nu G^{[k]}_\mu+...,
\end{equation}
where $E^{[k]}_{\mu\nu}$ is of hyperbolic type and $G^{[k]}_{\mu}$ are gauge fixing functions. The addition $\ldots$  comes from an additional conformal gauge fixing term that is basically of the form $g^{[0]}_{\mu\nu}\gamma^{[k]}$, for some additional (scale) gauge fixing functions $\gamma^{[k]}$. As in the standard method we will solve system $E^{[k]}_{\mu\nu}=0$. However, the system is still not strictly hyperbolic. Fortunately, it is some generalized type of hyperbolic equation for which we provide a proof of well-posedness in Section \ref{sec:proof-hip}. The standard method now use Bianchi identity to show that gauge fixing functions $G^{[k]}_\mu$ and $\gamma^{[k]}$ propagate by a linear hyperbolic system. 

Here the next problem appears. We use only part of the Ricci tensor from the ambient metric. The Bianchi identities are already used to deduce that the remaining parts $S^{[k]}_{\mu\infty}$ and $S^{[k]}_{\infty\infty}$ vanish. We denote by $\infty$ the ambient direction. In some way, the ambient metric is already in the partially gauge fixed form and additional gauge fixing is excessive. We circumvent this problem by building $G^{[k]}_{\mu}$ and $\gamma^{[k]}$  from these remaining parts of the ambient Ricci tensor in such a way that Bianchi identities provide generalized hyperbolic system for the gauge fixing functions (see Section \ref{sec:gauge}). 

Our method is more general and allows to prove well-posedness of various equations constructed with aid of the ambient construction. In particular, it is true for Graham-Jenne Mason-Sparling (GJMS) \cite{GJMS1} equation and its various generalizations. It is a linear system of the similar type as gauge fixed Anderson-Fefferman-Graham equation thus it has the unique development with a given initial data, which is global on globally hyperbolic spacetimes. Using this result, we show that there always exists a scale with vanishing Branson $Q$-curvature \cite{BransonQ} for a given globally hyperbolic spacetime. Another application is propagation of covariantly constant tractor \cite{Gover2004, Graham2011} from the initial Cauchy surface. Existence of the covariantly constant tractor is equivalent for metric to be conformal to Einsteinian metric, except when certain nondegeneracy condition is not satisfied what corresponds to conformal boundary (see \cite{Curry2015}). In this way one can show that the condition of being Einstein propagates to the whole development of AFG equation in a uniform way. The initial Cauchy surface can now also cross the conformal boundary (see Proposition \ref{prop:tractor}).

\section{Generalized hyperbolic systems}\label{sec:proof-hip}

We use abstract index notation and Einstein summation convention in the paper.
We denote indices in $M$ by Greek letters.  We use symbol $D^m u$ to denote $m$ jets on $M$ of the field $u$ on $M$.

We will consider a bit more general situation then a standard second order hyperbolic system. We consider a system involving family of multifields $u^{(k)}$ for $k=0,\ldots N$, where each multifield can be in other space. Consider a system of equations on $M$ for multifields $u^{(k)}$, $k=0,\cdots N$ 
\begin{equation}\label{eq:hiper}
K_k=-\frac{1}{2}g^{\iI\jI}(x,u^{(0)})\partial_{\iI}\partial_{\jI} u^{(k)}
+F_k\left(x,u^{(l)}, \partial_\mu u^{(l)}, \partial_\mu\partial_\nu u^{(l)}\right)=0,\quad k=0,\ldots, N,
\end{equation}
where functions $F_k(x,u^{(l)},v_\mu^{l)},w_{\mu\nu}^{(l)})$ depends smoothly on coordinates $x$ and
$u^{(l)}$ for $l\leq \max(k+1,N)$, $v_\mu^{(l)}$ for $l\leq k$, $w_{\mu\nu}^{(l)}$ for $l\leq k-1$. Function $g_{\mu\nu}(x,u^{(0)})$ is a lorentzian metric smoothly depending on $x$ and $u^{(0)}$.\footnote{These conditions should hold for some open neighbourhood of values of $u$ (for some open set in the standard Fr{\'e}chet topology on smooth sections on $\Sigma$). We will take this condition as obvious in what follows.}
We will call such system $K_n=0$ generalized hyperbolic for $u^{(k)}$, $k=0,\ldots N$.

We assume that $\Sigma\subset M$ is spacelike and compact. This is a condition for $u|_\Sigma$ if the metric is $u$ dependent. 
We consider smooth initial data
\begin{equation}\label{eq:data-hiper-C}
u^{(l)}|_\Sigma=f^{(l)}_0,\quad \partial_1 u^{(l)}|_\Sigma=f^{(l)}_1,\quad l=0,\ldots N.
\end{equation}
We are interested in well-posedness in smooth category. It means a series of important properties (see \cite{Ringstroem2009}):
\begin{enumerate}
\item \textbf{Existence of a unique local solution}: There exists $I=(-T_-,T_+)$, $T_\pm>0$ such that on $M=I\times \Sigma$ we have a unique and maximal smooth solution with the given initial data at $\{0\}\times \Sigma$.
Moreover, every surface $\{s\}\times \Sigma$ is a Cauchy surface with respect to the metric $g_{\mu\nu}(x,u^{(0)})$, thus $M$ is globally hyperbolic. 
\item \textbf{The speed of propagation is equal to the speed of light}: Namely, if two initial data $u,u'$ are equal on some opens set $U\subset \Sigma$ then the $u(x)=u'(x)$ in all points such that $[J^+_{g(u)}(x)\cup J^-_{g(u)}(x)]\cap \{0\}\times \Sigma\subset U$.
From this we can deduce some version of well-posedness also for arbitrary non-compact Cauchy surfaces.
\item {\bf Smooth dependence on the initial data}: For arbitrary $T_\pm'<T_\pm$ the solution is defined for an open neighbourhood of a given initial data. The solution depends smoothly on the initial data (as a map from a Fr{\'e}chet space of smooth sections on $\Sigma$ to a Fr{\'e}chet space of smooth sections on $M$). In particular, the derivative of the family of solutions satisfies a linearized equation.
\end{enumerate}
Let us first notice that the system is non-characteristic on every Cauchy surface. Suppose that $\Sigma=\{x^1=0\}$, so $x^1$ is a time function.

\begin{lm}\label{lm:weak-hip}
There exist smooth functions $L_k\left(\bar{D}^2u^{(l)}_{l\leq k}, \bar{D}\partial_1 u^{(l)}_{l\leq k}\right)$ valued in the multifields, such that the following conditions are equivalent
\begin{enumerate}
\item $K_k|_\Sigma=0$ for $k=0,\ldots N$
\item $\partial_1^2u^{(k)}|_\Sigma=L_k(\bar{D}^2u^{(l)}_{l\leq k}|_\Sigma, \bar{D}\partial_1 u^{(l)}_{l\leq k}|_\Sigma)$  for $k=0,\ldots N$,
\end{enumerate}
where $\bar{D}^n$ denotes $n$ jets on $\Sigma$.
\end{lm}

\begin{proof}
We will prove by induction in $k_0$ the statement:
\begin{equation}
\left(\forall k< k_0\colon K_k|_\Sigma=0\right)\Longleftrightarrow
\left(\forall k< k_0\colon\partial_1^2u^{(k)}|_\Sigma=L_k(\bar{D}^2u^{(l)}_{l\leq k}|_\Sigma, \bar{D}\partial_1 u^{(l)}_{l\leq k}|_\Sigma)\right).
\end{equation}
The statement for $k_0=0$ is tautological. Suppose that it is true for some $k_0\geq 0$.
Using $g^{11}\not=0$ we can write
\begin{equation}
-\frac{2}{g^{11}}K_{k_0+1}|_\Sigma=\partial_1^2u^{(k_0+1)}|_\Sigma+\ldots,
\end{equation}
where $\ldots$ is a smooth function of $\bar{D}^1\partial_1 u^{(k)}$ for $k\leq k_0+1$, $\bar{D}^2u^{(k)}$ for $k\leq k_0+1$ and $\partial_1^2 u^{(k)}$ for $k\leq k_0$. We can express $\partial_1^2 u^{(k)}$ for $k\leq k_0$ by $L_k$ by induction hypothesis, thus we obtain desired function finishing the induction proof.
\end{proof}

\begin{prop}\label{prop:hip}
The generalized hyperbolic system $K_n=0$ \eqref{eq:hiper} for the multifields $u^{(k)}$ for $k=0,\ldots N$ is well-posed in the smooth category. If the system is linear (or linear with a source term) then the solution is defined on the whole globally hyperbolic spacetime.
\end{prop}

\begin{proof}
We will prove Proposition \ref{prop:hip} by induction with respect to the order $N$. For $N=0$ it is a known result (see for example \cite{taylor} Chapter 16.1-16.3 and \cite{Choquet-Bruhat} Appendix III). The following system is well-posed:
\begin{equation}\label{eq:wave}
-\frac{1}{2}g^{\mu\nu}(x,u)\partial_\mu\partial_\nu u+F(x, D^1u)=0,
\end{equation}
where $g_{\mu\nu}(x,u)$ is a lorentzian metric smoothly depending on coordinates $x$ and $u$ and $F$ is a smooth function of  coordinates and $D^1u$ (first jets of $u$). 

We assume now that we proved the statement for all $0\leq N<N_0$. Consider generalized hyperbolic system with $N_0$ multifields 
\begin{equation}
-\frac{1}{2}g^{\iI\jI}(x,u^{(0)})\partial_{\iI}\partial_{\jI} u^{(k)}
+F_k\left(x,u^{(l)}, \partial_\mu u^{(l)}, \partial_\mu\partial_\nu u^{(l)}\right)=0,\quad k=0,\ldots N_0,
\end{equation}
for some functions $F_k(x,u^{(l)},v_\mu^{l)},w_{\mu\nu}^{(l)})$ depending on variables described in the definition of generalized hyperbolic system.
Differentiating equation for $u^{(k)}$ with respect to $\partial_\rho$ we get
\begin{align}
-\frac{1}{2}g^{\mu\nu}\partial_\mu\partial_\nu \partial_\rho u^{(k)}&-\frac{1}{2}
\frac{\partial g^{\mu\nu}}{\partial u^{(0)}}\partial_\rho u^{(0)}\partial_\mu \partial_\nu u^{(k)}-\frac{1}{2}\frac{\partial g^{\iI\jI}}{\partial x^\rho}\partial_{\iI}\partial_{\jI} u^{(k)}+\nonumber\\
&+\frac{\partial F_k}{\partial u^{(l)}}\partial_\rho u^{(l)}+\frac{\partial F_k}{\partial v_\mu^{(l)}}\partial_\mu \partial_\rho u^{(l)}+\frac{\partial F_k}{\partial w_{\mu\nu}^{(l)}}\partial_\mu\partial_\nu \partial_\rho u^{(l)}
+\frac{\partial F_k}{\partial x^\rho}=0,
\end{align}
where summation over $l$ is implicitely assumed (as well as standard Einstein summation convention).
Introducing $p_\mu^{(k)}=\partial_\mu u^{(k)}$ for $k\leq N_0-1$ we can write it as
\begin{equation}
-\frac{1}{2}g^{\mu\nu}\partial_\mu\partial_\nu p_\rho^{(k)}+G_{\rho,k}\left(x,u^{(l)}, \partial_\mu u^{(l)}, \partial_\mu\partial_\nu u^{(l)}, p_\rho^{(l)}, \partial_\mu p_\rho^{(l)}, \partial_\mu\partial_\nu p_\rho^{(l)}\right)=0,
\end{equation}
where
\begin{align}
&G_{\rho,k}(x,u^{(l)},v_\mu^{(l)},w_{\mu\nu}^{(l)},p_\rho^{(l)},q_{\mu\rho}^{(l)},s_{\mu\nu\rho}^{(l)})=\nonumber\\
&=-\frac{1}{2}\frac{\partial g^{\mu\nu}}{\partial u^{(0)}}p_\rho^{(0)}q^{(k)}_{\mu\nu}+\frac{\partial F_k}{\partial u^{(l)}}p_\rho^{(l)}+\frac{\partial F_k}{\partial v_\mu^{(l)}}q_{\mu\rho}^{(l)}+\frac{\partial F_k}{\partial w_{\mu\nu}^{(l)}}s_{\mu\nu\rho}^{(l)}-\frac{1}{2}\frac{\partial g^{\iI\jI}}{\partial x^\rho}q^{(k)}_{\iI\jI}+\frac{\partial F_k}{\partial x^\rho},
\end{align}
where all derivatives of $F_k$ and $g_{\mu\nu}$ retain their original variables.
We introduce new multifields
\begin{align}\label{eq:fields-new}
{u'}^{(k)}=\{u^{(k)}, p_\mu^{(k)}\}\text{ for } k< N_0-1,\\
{u'}^{(N_0-1)}=\{u^{(N_0-1)}, p_\mu^{(N_0-1)},u^{(N_0)}\},
\end{align}
and the system of equations
\begin{align}
&-\frac{1}{2}g^{\mu\nu}(x,u^{(0)})\partial_\mu\partial_\nu p_\rho^{(k)}+G_{\rho,k}\left(x,u^{(l)}, \partial_\mu u^{(l)}, \partial_\mu\partial_\nu u^{(l)}, p_\rho^{(l)}, \partial_\mu p_\rho^{(l)}, \partial_\mu\partial_\nu p_\rho^{(l)}\right)=0,\quad k\leq N_0-1,\label{eq:p}\\
&-\frac{1}{2}g^{\iI\jI}(x,u^{(0)})\partial_{\iI}\partial_{\jI} u^{(k)}
+F_k\left(x,u^{(l)}, \partial_\mu u^{(l)}, \partial_\mu \partial_\nu u^{(l)}\right)=0,\quad k\leq N_0-1,\label{eq:u}\\
&-\frac{1}{2}g^{\iI\jI}(x,u^{(0)})\partial_{\iI}\partial_{\jI} u^{(N_0)}
+F_{N_0}\left(x,u^{(l)}, \partial_\mu u^{(l)}, \partial_\mu p_\nu^{(l)}\right)=0.\label{eq:u-last}
\end{align}
It is of the generalized form but the order is now $N_0-1$.

From solution of original system we can form solution of this system by taking
\begin{equation}
p_\mu^{(k)}=\partial_\mu u^{(k)}.
\end{equation}
We proved uniqueness by induction hypothesis.

In order to prove existence we construct initial data
\begin{equation}
p_\mu^{(k)}|_\Sigma=\partial_\mu u^{(k)}|_\Sigma,\quad \partial_1 p_\mu^{(k)}|_\Sigma=\partial_1\partial_\mu u^{(k)}|_\Sigma,
\end{equation}
where $\partial_1\partial_\mu u^{(k)}|_\Sigma$ we compute from Lemma \ref{lm:weak-hip} (the original system is non-characteristic with respect to a Cauchy surface, thus it is possible).

Now we notice that the difference $\partial_\rho \eqref{eq:u}-\eqref{eq:p}$ is equal
\begin{align}\label{eq:wmuk}
-\frac{1}{2}g^{\iI\jI}\partial_\mu\partial_\nu w_\mu^{(k)}&-\frac{1}{2}\frac{\partial g^{\mu\nu}}{\partial u^{(0)}}w_\rho^{(0)}\partial_\mu \partial_\nu u^{(k)}_\nu -\frac{1}{2}\frac{\partial g^{\mu\nu}}{\partial u^{(0)}}p_\rho^{(0)}\partial_\mu w^{(k)}_\nu-\frac{1}{2}\frac{\partial g^{\iI\jI}}{\partial x^\rho}\partial_{\iI}w^{(k)}_{\jI}+\nonumber\\
&+\frac{\partial F_k}{\partial u^{(l)}}w_\rho^{(l)}+\frac{\partial F_k}{\partial v_\mu^{(l)}}\partial_\mu w_\rho^{(l)}+\frac{\partial F_k}{\partial w_{\mu\nu}^{(l)}}\partial_\mu\partial_\nu w_\rho^{(l)}=0,\quad k\leq N_0-1,
\end{align}
where $w_\mu^{(k)}=\partial_\mu u^{(k)}-p_\mu^{(k)}$, $k=0,\ldots N_0-1$. These equations form linear generalized hyperbolic system for $w_\mu^{(k)}$ and as the initial data 
\begin{equation}
w_\mu^{(k)}|_\Sigma=\partial_\mu u^{(k)}-p_\mu^{(k)}|_\Sigma=0,\quad 
\partial_1w_\mu^{(k)}|_\Sigma=\partial_1\partial_\mu u^{(k)}-\partial_1p_\mu^{(k)}|_\Sigma=0,
\end{equation}
we have by uniqueness of solution (due to induction hypothesis)
\begin{equation}
p_\mu^{(k)}=\partial_\mu u^{(k)},
\end{equation}
and the solution of the lower order system gives the solution of the original one. 

The solution of $u'$ system depends smoothly on the initial data, so it is also true for the original system. The induction is complete.
\end{proof}

\begin{remark}
In fact one can show that it is well-posed in Sobolev spaces, but of different order for every $u^{(k)}$. We leave the details for further investigations.
\end{remark}

\subsection{The ambient construction}

We consider an ambient space\footnote{We will consider later another ambient space $\mM=\R\times\tilde{M}$ which was introduced by Fefferman and Graham \cite{Fefferman-Graham}.}
\begin{equation}
\tilde{M}= M\times \R.
\end{equation}
with coordinates  $(x^{\iI},\rho)$ where $x^{\iI}$ are coordinates on $M$. 
We will denote fields on $\tilde{M}$ with $\tilde{\phantom{\phi}}$. We regard them as a formal series in $\rho$. 

We denote differentiation over $\rho$ by $\partial_\infty$ or $'$. We denote indices in $M$ by greek letters. We assume that $x^1$ is a time coordinate and in what follows $\Sigma=\{x^1=0\}\subset M$. We use symbol $D^m u$ to denote $m$ jets on $M$ of the field $u$ on $M$.

Let us consider a multifield (a collection of fields) $\tilde{v}$ on $\tilde{M}$, we can write an expansion
\begin{equation}
\tilde{v}=\sum_{m=0} \tilde{v}^{[m]}(x^{\iI})\rho^m + O(\rho^\infty),
\end{equation}
where $\tilde{v}^{[m]}$ are rescaled Taylor expansion coefficients and $O(\rho^\infty)$ means a term that vanishes to infinite order at $\rho=0$ surface. Every term in the expansion is a function on $M$. In what follows we will be interested in such formal series.

\begin{df}\label{df:order}
Let $\tilde{u}$ be multifield on $\tilde{M}$.
We say that a formal series $\tilde{F}$ in $\rho$
is of order $N$ in $\tilde{u}$ if $\tilde{F}^{[n]}$ is a function of $x$ and
\begin{equation}
\left\{D^{m}\tilde{u}^{[l]}\right\},\quad m=\min\{n+N-l,2\},\ l\leq n+N.
\end{equation}
\end{df}

Let $\tilde{D}^m \tilde{u}$ denotes $m$ jets on $\tilde{M}$ of the field $\tilde{u}$. If
$\tilde{F}(x,\tilde{D}^2\tilde{u})$ is a smooth function then it is of order $2$.
Let $\tilde{f}^{\mu\nu}(x,\tilde{u})$  be a  smooth tensor function then
\begin{equation}\label{eq:g-exp}
\tilde{f}^{\mu\nu}\partial_\mu\partial_\nu \tilde{u}=[\tilde{f}^{\mu\nu}]^{[0]}\partial_\mu\partial_\nu \tilde{u}+\tilde{F},
\end{equation}
where $\tilde{F}$ is of order $1$.

Important example of generalized hyperbolic systems can be obtained from $\tilde{K}^{[n]}=0$, $n=0,\ldots, N$ for multifields $\tilde{u}^{[n]}$, $n=0,\ldots, N$ if
\begin{equation}\label{eq:ambient-hip}
\tilde{K}=-\frac{1}{2}\g^{\iI\jI}\partial_{\iI}\partial_{\jI}\tilde{u}+\tilde{F},
\end{equation}
and $\tilde{F}$ is of order $1$ and $\tilde{F}^{[n]}$ for $n\leq N$ decouple in the sense that they does not depend on $O(\rho^{N+1})$ part of the multifield.

\subsection{The derived equation}

The equations of interest have also another important property:

\begin{df}\label{df:recursive}
We say that a system $K_n(x, D^2u)=0$, $n=0,\ldots, N$ for $u^{(k)}$, $k=0,\ldots, N$ is recursive if 
\begin{enumerate}
\item For every $n<N$, $K_n$ is a function of $D^2u^{(k)}$ for $k\leq n$ and $u^{(n+1)}$,
\item For every $n<N$, $K_n$ depends linearly on $u^{(n+1)}$ and we can determine 
$u^{(n+1)}$ from equation $K_n=0$ in terms of other variables.
\end{enumerate}
\end{df}

We will consider in this paper the generalized hyperbolic systems given by
\begin{equation}\label{eq:long}
\tilde{K}^{[k]}=-\frac{1}{2}g^{\mu\nu}(\tilde{u}^{[0]})\partial_\mu\partial_\nu \tilde{u}^{[k]}+[\tilde{F}]^{[k]},
\end{equation}
which are recursive till order $N$. In order to determine this property it is enough to study $\tilde{F}$.

The property \eqref{df:recursive} allows us to determine higher order variables from sufficiently high jets of the lowest order $u^{(0)}$. In our application we need a local version of this procedure that is described by a following lemma:

\begin{lm}\label{lm:recursive}
Let $K_n$ be recursive in $u^{(k)}$ for $0\leq k\leq N$ till order $N$.
There exist smooth functions $H^n_{\tilde{K}}$ for $0<n\leq N$ depending on $x\in M$ and on variables $D^{2n}u^{(0)}$,
such that for any point $x\in M$ and an integer $N'>0$ the following conditions are equivalent
\begin{enumerate}
\item $D^{N'-2k-2}K_k(x)=0$ for $0\leq k\leq N-1$, $2k+2\leq N'$,
\item $D^{N'-2k}\left(u^{(k)}-H_{\tilde{K}}^k\left(x,D^{2k}u^{(0)}\right)\right)(x)=0$ for $1\leq k\leq N$, $2k\leq N'$.
\end{enumerate}
where $D^m$ denotes $m$-th jets. In the case of linear system the functions $H^n_{\tilde{K}}$ are also linear. If the system does not directly depend on $x$ then the same is true for $H^n_{\tilde{K}}$.
\end{lm}

\begin{remark}
We will use subscript for $H^k$ to indicate which system is used to determine recursive functions.
\end{remark}

\begin{proof}
We proceed by induction in $k_0$. Suppose that 
\begin{equation}
D^{N'-2k}u^{(k)}(x)=D^{N'-2k}H_{\tilde{K}}^k\left(x,D^{2k}u^{(0)}\right)|_x,\quad 1\leq k\leq  k_0-1
\end{equation}
is equivalent to
\begin{equation}
D^{N'-2k-2} K_k(x)=0,\quad 0\leq k\leq k_0-2.
\end{equation}
Solving equation $K_{k_0-1}=0$ for $u^{(k_0)}$ we introduce functions $G^{k_0}$
\begin{equation}
u^{(k_0)}=G^{k_0}\left(x,\left\{D^{2}u^{(l)}\right\}_{ l\leq k_0-1}\right)
\end{equation}
Let us notice that for $k_0\leq N$ and $2k_0\leq N'$
\begin{equation}
 D^{N'-2k_0}K_{k_0-1}(x)=0,\Longleftrightarrow  D^{N'-2k_0}(u^{(k_0)}-G^{k_0})|_x=0,
\end{equation}
due to linear dependence.
By inserting recursively variables from the lower orders we show the result. Induction starts with $k_0=0$ where it is a trivial statement.
\end{proof}

This lemma allows us to determine initial data for the generalized hyperbolic system from the sufficienty high jets of the lowest order field $u^{(0)}$ on the Cauchy surface. Important is that the evolved generalized system will have $u^{(0)}$ in agreement with this data. This property is guaranteed by the following fact:

\begin{lm}\label{lm:bijective}
Let $K_n=0$ be  the generalized hyperbolic system, recursive in the multifield $u^{(k)}$, $k=0,\ldots N$. Then a solution with initial data
\begin{equation}
u^{(k)}|_\Sigma=H_{\tilde{K}}^k\left(\cdot,D^{2k}v|_\Sigma\right),\quad \partial_1u^{(k)}|_\Sigma=\partial_1H_{\tilde{K}}^k\left(\cdot,D^{2k}v|_\Sigma\right),
\end{equation}
on a Cauchy surface satisfies
\begin{equation}
D^{2N+1}u^{(0)}|_\Sigma=D^{2N+1}v|_\Sigma.
\end{equation}
\end{lm}

\begin{proof}
Let $u^{(k)}$ be a development. We denote
\begin{equation}
\partial_1^mv^{(k)}|_\Sigma=\partial_1^mH_{\tilde{K}}^k\left(\left\{D^{m}v|_\Sigma\right\}_{m\leq 2k}\right),\quad m+2k\leq 2N+1.
\end{equation}
Consider a set
\begin{equation}
A=\{(m,k)\in \Z_+\times \Z_+\colon m\geq 0,\ k\geq 0,\ m+2k\leq 2N+1,\ \partial_1^mv^{(k)}|_\Sigma\not=\partial_1^mu^{(k)}|_\Sigma\}.
\end{equation}
We should show that this set is empty. By contradition assume otherwise and define
\begin{equation}
m_0=\min(m\colon \exists k,\ (m,k)\in A),\quad k_0=\min(k\colon (m_0,k)\in A).
\end{equation}
We notice that $m_0\geq 2$ because of the way $u^{(k)}|_\Sigma$ and $\partial_1u^{(k)}|_\Sigma$ are defined.
Consider $\partial_1^{m_0-2}K_{k_0}$ in terms of $u^{(k_0)}$
\begin{equation}
\partial_1^{m_0-2}K_{k_0}|_\Sigma=-\frac{1}{2}g^{11}\partial_1^{m_0}u^{(k_0)}|_\Sigma+\ldots,
\end{equation}
where $\ldots$ is a function of the terms which do not belong to $A$ by definition. Similarly
\begin{equation}
\partial_1^{m_0-2}K_{k_0}|_\Sigma=-\frac{1}{2}g^{11}\partial_1^{m_0}v^{(k_0)}|_\Sigma+\ldots,
\end{equation}
and as $\ldots$ of the same property. From the definition of $(m_0,k_0)$ the remainders $\ldots$ are equal thus
\begin{equation}
-\frac{1}{2}g^{11}\partial_1^{m_0}u^{(k_0)}|_\Sigma=-\frac{1}{2}g^{11}\partial_1^{m_0}v^{(k_0)}|_\Sigma,
\end{equation}
and as $g^{11}|_\Sigma$ is nonvanishing we obtain a contradition.
\end{proof}

If the equation system $K_n=0$ is recursive till order $N$ and it decouples at this order,
then the equation for $n=N$ gives us by Lemma \ref{lm:recursive} the equation of higher order for $u^{(0)}$
\begin{equation}\label{eq:G}
0=K_N\left(x,\left\{D^{2}u^{(l)}\right\}_{l\leq N}\right)=H^{N+1}_{\tilde{K}}\left(x,D^{2N+2}u^{(0)}\right).
\end{equation}
We will called it the derived equation from the system. In case of a linear system the derived equation is also linear. 

If the system $K_n=0$ , $n=0,\ldots , N$ is satisfied, then the equation \eqref{eq:G} for $u^{(0)}$ also holds. From solution of \eqref{eq:G} we can obtain solution to the system. Initial conditions for this equation provide also initial conditions for the system, if we know sufficiently high jets on the initial surface.

\section{The Fefferman-Graham ambient metric construction}\label{sec:ambient}

We are working in even $d$ dimensions. Moreover we assume that $d\geq 4$. Let us introduce an ambient space $\mM$ for the spacetime $M$
\begin{equation}
\mM=\R_+\times \tilde{M},\quad \tilde{M}= M\times \R,
\end{equation}
with coordinates $(t, x^{\iI},\rho)$ and $(x^{\iI},\rho)$ respectively, where $x^{\iI}$ are coordinates on $M$ and the metric on $\mM$ takes the form
\begin{equation}
\mg_{IJ} dx^Idx^J=2\rho dt^2+2tdtd\rho+t^2\tilde{g}_{\iI\jI}(x^{\iI},\rho)dx^{\iI}dx^{\jI}.
\end{equation}
In the following we will denote objects on $\tilde{M}$ with $\tilde{\phantom{\phi}}$ and objects on $\mM$ bold. We denote by $\mg_{IJ}$, $\mnabla_{I}$, $\mRS_{IJ}$ metric covariant derivative and Ricci tensor respectively on $\mM$. Indices $I=0,\infty$ or $\iI$ in the case of index on $M$. We use $\mg_{IJ}$ to raise or lower indices. The metric $\g_{\iI\jI}$, connection $\tnabla_{\iI}$ and Ricci tensor $\RR_{\iI\jI}$ are depending on $\rho$ objects on $M$. We use $\g_{\iI\jI}$ to raise and lower indices for such objects.

Let $h_{\iI\jI}$ be a given metric on $M$. The ambient metric on $\mM$ is a metric $\rm$that satisfies
$\g_{\iI\jI}^{[0]}=h_{\iI\jI}$ and
\begin{equation}\label{eq:FG}
\mRS_{IJ}=O(\rho^{d/2-1}),\quad \mRS=O(\rho^{d/2}),
\end{equation}
where $\mRS_{IJ}$ and $\mRS$ are Ricci tensor and Ricci scalar of the metric on $\mM$.
Symbol $\mF=O(\rho^n)$ means that  $\lim_{\rho\rightarrow 0}\rho^{-n}\mF$ exists.

One can show that $\mRS_{0 I}=0$ and that $\mRS_{IJ}$ is $t$ independent.   Essentially, it is a function on $\tilde{M}$ (see \cite{Fefferman-Graham})
\begin{equation}
\mRS_{\iI\jI}:=\RS_{\iI\jI},\quad \mRS_{\iI\infty}:=\RS_{\iI\infty},\quad \mRS_{\infty\infty}:=\RS_{\infty\infty}.
\end{equation}
We have  (eq. 3.17 in \cite{Fefferman-Graham})
\begin{align}
&\RS_{\iI\jI}=\rho \g_{\iI\jI}''-\rho \g^{\kI\lI}\g_{\kI\iI}'\g_{\lI\jI}'-\left(\frac{d}{2}-1\right)\g_{\iI\jI}'-\frac{1}{2}\g^{\kI\lI}\g_{\kI\lI}'\g_{\iI\jI}+\RR_{\iI\jI},\label{eq:S-form}\\
&\RS_{\iI\infty}=\frac{1}{2}\g^{\kI\lI}\left(\tnabla_{\kI}\g_{\iI\lI}'-\tnabla_{\iI}\g_{\kI\lI}'\right),\\
&\RS_{\infty\infty}=-\frac{1}{2}\g^{\kI\lI}\g_{\kI\lI}''+\frac{1}{4}\g^{\kI\lI}\g^{\iI\jI}\g_{\iI\kI}'\g_{\jI\lI}',\label{eq:S-infty2}
\end{align}
where $\RR_{\iI\jI}$ denotes the Ricci tensor in the metric $\g_{\iI\jI}$ depending on $\rho$.
The equations \eqref{eq:FG} are equivalent to
\begin{align}
&\RS_{\iI\jI}^{[n]}=0,\quad n=0,\ldots d/2-2,\label{eq:FG-2}\\
&(\g^{[0]})^{\iI\jI}\RS_{\iI\jI}^{[d/2-1]}=0,
\end{align}
and then other components automatically vanish. Namely, (see \cite{Fefferman-Graham} and compare with Proposition \ref{prop:Bianchi}),
\begin{equation}\label{eq:Bianchi-FG}
\RS_{\iI\infty}=O(\rho^{d/2-1}),\quad\RS_{\infty\infty}=O(\rho^{d/2-1}).
\end{equation}
One can check that $\RS_{\iI\jI}^{[n]}$ is recursive till order $d/2-1$. Thus, we obtain \footnote{This recurrence procedure breaks down for $n=d/2-1$ and this is the source of the obstruction tensor.}
\begin{equation}\label{eq:g-S-iter}
\g_{\iI\jI}^{[n]}=H_{\iI\jI\RS}^n(D^{2n}\g^{[0]}),\quad n\leq d/2-1,
\end{equation}
so higher orders of the metric are determined through $\g_{\iI\jI}^{[0]}=h_{\iI\jI}$.
The last equation $(\g^{[0]})^{\iI\jI}\RS_{\iI\jI}^{[d/2-1]}=0$ allows to compute the trace $\operatorname{tr} \g^{[d/2]}=(\g^{[0]})^{\iI\jI}\g_{\iI\jI}^{[d/2]}$. 
The formula for $\RS_{\iI\jI}^{[d/2-1]}$ depends on $\g^{[d/2]}$ only through the trace. It does not depend on the choice of the ambient metric. 

The Fefferman-Graham obstruction tensor for $h_{\iI\jI}$ is defined by
\begin{equation}
H_{\iI\jI}=\RS_{\iI\jI}^{[d/2-1]}.
\end{equation}
The constraints $H_{\iI\jI}=0$ are equivalent to 
\begin{equation}
\RS_{\iI\jI}^{[n]}=0,\quad n=0,\ldots d/2-1.
\end{equation}
Let us notice, that the specific combination
\begin{equation}\label{eq:comb}
\RS_{\iI\jI}^{[d/2-1]}-\frac{1}{d/2-1}\g_{\iI\jI}^{[0]}\RS_{\infty\infty}^{[d/2-2]}
\end{equation}
depends only on $\g^{[k]}$ for $k\leq d/2-1$ and its derivatives. Importantly,
\begin{equation}\label{eq:obstruction}
H_{\iI\jI}=\RS_{\iI\jI}^{[d/2-1]}-\frac{1}{d/2-1}\g_{\iI\jI}^{[0]}\RS_{\infty\infty}^{[d/2-2]},
\end{equation}
because $\RS_{\infty\infty}^{[d/2-2]}=0$ by \eqref{eq:Bianchi-FG}.
This form of the obstruction tensor does not involve $\operatorname{tr} \g^{[d/2]}$.

\subsection{Hyperbolicity of the Anderson-Fefferman-Graham equation}

The system of \eqref{eq:FG-1} is not hyperbolic for the same reason as Einstein's gravity, because of the gauge transformations. The first step is to introduce a hyperbolic system in a specific gauge. We will use a natural gauge introduced in \cite{Anderson}. Then we show that this gauge is preserved in the evolution and as a result the obtained solution is also a solution of the Anderson-Fefferman-Graham equations.  It is a standard treatment in gravity (see \cite{Choquet-Bruhat} for application to Einstein's equations).

Let us remind the following  known identity (see \cite{Choquet-Bruhat})
\begin{equation}\label{eq:R-property}
\RR_{\iI\jI}=-\frac{1}{2}\g^{\kI\lI}\partial_{\kI}\partial_{\lI} \g_{\iI\jI}+\frac{1}{2}\left(\partial_{\iI} \tF_{\jI}+\partial_{\jI}\tF_{\iI}\right)+\ldots,
\end{equation}
where $...$ means terms of order $1$ (see Definition \ref{df:order}) and
\begin{equation}
\tF_{\iI}=\g^{\kI\lI}\left(\partial_{\kI} \g_{\lI\iI}-\frac{1}{2}\partial_{\iI} \g_{\kI\lI}\right).
\end{equation}
We notice that
\begin{equation}
\partial_\infty^2 \tF_{\iI}=\g^{\kI\lI}\left(\partial_{\kI} \g_{\lI\iI}''-\frac{1}{2}\partial_{\iI} \g_{\kI\lI}''\right)+\ldots,
\end{equation}
where $\ldots$ denotes terms of order at most $2$.
Comparing it with
\begin{align}
&\partial_\infty \RS_{\iI\infty}=\frac{1}{2}\g^{\kI\lI}\left(\partial_{\kI}\g_{il}''-\partial_{\iI}\g_{\kI\lI}''\right)+\ldots,\\
&\partial_{\iI}\RS_{\infty\infty}=-\frac{1}{2}\g^{\kI\lI}\partial_{\iI}\g_{\kI\lI}''+\ldots,
\end{align}
we obtain the formula
\begin{equation}\label{eq:tF-tG}
\partial_\infty^2 \tF_{\iI}=2\partial_\infty\RS_{\iI\infty}-\partial_{\iI}\RS_{\infty\infty}+\ldots,
\end{equation}
where $\ldots$ denotes terms of order at most $2$. 

In order to write a slightly modified $\tF_{\iI}$ in terms of $\RS_{\iI\infty}$ and $\RS_{\infty\infty}$ we
extend the notion of derivatives with respect to $\rho$. For $n>0$ we introduce $n$-times integration of a multifield $\tilde{u}$ (a collection of fields on $\tilde{M}$) 
\begin{equation}
\partial_\infty^{-n}\tilde{u}(x^{\iI},\rho)=\int_0^\rho d\rho'\ \frac{(\rho-\rho')^{n-1}}{(n-1)!}\tilde{u}(x^{\iI},\rho'),
\end{equation}
that is $\partial_\infty^{-n} \sum_{k=0} u^{[k]}\rho^k=\sum_{k=0} \frac{1}{(k+1)\cdots (k+n)}u^{[k]}\rho^{k+n}$.
Suppose that $\tilde{F}$ is of order $N$ then $\partial_\infty^{-n}\tilde{F}$ is of order $N-n$.

We introduce additional tensors
\begin{align}
\tS&=-\frac{1}{2}\g^{[0]\kI\lI}\g_{\kI\lI}^{[1]}+\partial_\infty^{-1}\RS_{\infty\infty},\label{eq:tS}\\
\tG_{\iI}&=\tF^{[0]}_\mu+2\partial_\infty^{-1}\RS_{\iI\infty}-\partial_{\iI}\partial^{-1}_\infty\tS,\label{eq:tG}\\
\tE_{\iI\jI}&=\RS_{\iI\jI}-\frac{1}{2}(\tnabla_{\iI}\tG_{\jI}+\tnabla_{\jI}\tG_{\iI})-\g_{\iI\jI}\tS,\label{eq:tE}
\end{align}
These tensors will be used in our analysis of the AFG equations.  The reason for occurence of additional term $\g_{\iI\jI}\tS$ is explained in the proof below.

\begin{lm}\label{lm:E-weel-posed}
The equation system 
\begin{equation}
\label{eq:Hip}
\tE_{\iI\jI}=O(\rho^{d/2})
\end{equation}
is generalized hyperbolic and recursive in $\g_{\iI\jI}^{[n]}$, $n=0,\ldots, d/2-1$ till order $d/2-1$ and it decouples at this order. Thus, it is well-posed.
\end{lm}

\begin{proof}
We will first prove that it is recursive till order $d/2-1$ and that it decouples at this order.
Functions $\tG_{\iI}$ and $\tS$ are of order $1$. Hence, $\tE_{\iI\jI}$ is of order $2$. Moreover, $\tG_{\iI}^{[n]}$ does not depend on $\g_{\iI\jI}^{[k]}$, $k>n$ (nor its derivatives). The dependence of $\RS_{\iI\jI}^{[n]}$ and $\tS^{[n]}$ on $D^m\g_{\iI\jI}^{[k]}$, $k>n$ is only by linear terms in $\g_{\iI\jI}^{[n+1]}$. In fact, we can compute  
\begin{equation}
[\RS_{\iI\jI}-\g_{\iI\jI}\tS]^{[n]}=\left(n-\frac{d}{2}+1\right)(n+1)\g_{\iI\jI}^{[n+1]}+\ldots,
\end{equation}
where $\ldots$ are terms depending on $D^2\g_{\iI\jI}^{[l]}$ for $l\leq n$.
For $n<d/2-1$ we can uniquely determined $\g_{\iI\jI}^{[n+1]}$. Additionally, $[\RS_{\iI\jI}-\g_{\iI\jI}\tS]^{[d/2-1]}$ depends only on $\g_{\iI\jI}^{[k]}$ for $k\leq d/2-1$ and their derivatives (see also \eqref{eq:comb}). Thus, the system is recursive and it decouples.

We need to show that $\tE_{\iI\jI}$ is of the form \eqref{eq:ambient-hip}.
As a preliminary step we prove that
\begin{equation}\label{eq:F-G}
\tF_{\jI}=\tG_{\jI}+\ldots,
\end{equation}
where $\ldots$ denotes term of order $0$. Indeed,
$\tF_{\jI}^{[0]}=\tG_{\jI}^{[0]}$. Direct computation gives
\begin{equation}
\tG_{\iI}^{[1]}=[\partial_\infty \tG_{\iI}]^{[0]}=2\RS_{\iI\infty}^{[0]}-\partial_{\iI}\tS^{[0]},
\end{equation}
which can be compared to
\begin{equation}
\tF_{\iI}^{[1]}=2\RS_{\iI\infty}^{[0]}+\frac{1}{2}\partial_{\iI}(\g^{[0]\kI\lI}\g_{\kI\lI}^{[1]})+\ldots,
\end{equation}
where $\ldots$ denotes terms depending on $x$ and  $D^m\g_{\iI\jI}^{[k]}$ for $m+k\leq 1$.
Finally, by \eqref{eq:tF-tG} we obtain
\begin{equation}
\partial_\infty^2\tF_{\jI}=\partial_\infty^2\tG_{\jI}+\ldots,
\end{equation}
where $\ldots$ denotes term of order $2$. This shows \eqref{eq:F-G}.

We thus have
\begin{equation}
\tnabla_{\iI} \tF_{\jI}=\tnabla_{\iI}\tG_{\jI}+\ldots,
\end{equation}
where $\ldots$ denotes term of order $1$ (both $\tF_{\iI}$ and $\tG_{\iI}$ depends only on up to first derivatives of the metric). 
Taking this   and  \eqref{eq:R-property} into account, the following yields
\begin{equation}
\tE_{\iI\jI}=\RS_{\iI\jI}-\frac{1}{2}(\tnabla_{\iI}\tG_{\jI}+\tnabla_{\jI}\tG_{\iI})-\g_{\iI\jI}\tS=\RR_{\iI\jI}-\frac{1}{2}(\tnabla_{\iI}\tF_{\jI}+\tnabla_{\jI}\tF_{\iI})+\ldots=
-\frac{1}{2}\g^{\kI\lI}\partial_{\kI}\partial_{\lI} \g_{\iI\jI}+\ldots,
\end{equation}
where $\ldots$ is of order $1$.
Expanding first term the form described in \eqref{eq:ambient-hip} is obtained. 
Well-posedness follows from Proposition \ref{prop:hip}.
\end{proof}

\subsubsection{Propagation of the gauge}\label{sec:gauge}

In this section we will explain that $\tG_{\iI}=O(\rho^{d/2})$ and $\tS=O(\rho^{d/2-1})$ provided that these functions vanish on the initial surface together with their time derivatives and secondly $\tE=O(\rho^{d/2})$. As usual this is achieved by showing that these variables obey a system of linear hyperbolic equations.

Let us introduce two tensors
\begin{align}
&\tB^1_{\iI}=-\frac{1}{2}\tnabla^{\kI}\tnabla_{\kI} \tG_{\iI}-\frac{1}{2}\RR_{\iI}^{\jI}\tG_{\jI}-\left(\frac{d}{2}-1-\rho\partial_\infty\right)\partial_\infty \tG_{\iI}+\frac{1}{2}\g^{\kI\lI}\g'_{\kI\lI}\rho\partial_\infty \tG_{\iI}+\frac{1}{2}\rho\g^{\kI\lI}\g'_{\kI\lI}\partial_{\iI}\tS,\\
&\tB^2=-\frac{1}{2}\tnabla^{\iI}\partial_{\iI}\tS
-\left(\frac{d}{2}-2-\rho\partial_\infty\right)\partial_\infty\tS+\g^{\iI\jI}\g_{\iI\jI}'\rho\partial_\infty\tS
+\frac{1}{2}\tQ^{\iI}\tG_{\iI}+\frac{1}{2}\g_{\iI\jI}'\tnabla^{\iI} \tG^{\jI}+\frac{1}{2}\g_{\iI\jI}'\g^{\iI\jI}\tS,
\end{align}
where $\tQ^{\iI}=\partial_\infty (\g^{\jI\kI}\tilde{\Gamma}^{\iI}_{\jI\kI})$  and $\tilde{\Gamma}^{\iI}_{\kI\lI}$ is $\rho$-dependent Christoffel symbol.

We will prove in Proposition \ref{lm:B-follows} in Section \ref{sec:Bianchi} that
if $\tE=O(\rho^{d/2})$ then $\tB^1_{\iI}=O(\rho^{d/2})$ and $\tB^2=O(\rho^{d/2-1})$. The detailed property is presented below:

\begin{lm}\label{lm:B-closed}
The equation system
\begin{equation}\label{eq:BB}
\tB^1_{\iI}=O(\rho^{d/2}),\ \tB^2=O(\rho^{d/2-1})
\end{equation}
is linear generalized hyperbolic for $\tG_{\iI}^{[n]}$, $n=0,\ldots d/2-1$ and $\tS^{[n]}$, $n=0,\ldots d/2-2$. Moreover, if \eqref{eq:BB} is satisfied and in a point $x\in M$ the following equations are true
\begin{equation}
D^{d-2} \tG_{\iI}^{[0]}(x)=0,\quad D^{d-3} \tS^{[0]}(x)=0,
\end{equation}
then 
\begin{equation}
D^{d-2k-2}\tG_{\iI}^{[k]}(x)=0,\ k=0,\ldots d/2-1,\quad D^{d-2k-3} \tS^{[k]}(x)=0,\ k=0,\ldots d/2-2.
\end{equation}
\end{lm}

\begin{proof}
Inspection of the equations shows that
\begin{equation}
\tB^1_{\iI}=-\frac{1}{2}g^{\mu\nu}\partial_\mu\partial_\nu \tG_{\iI}+\tF_{\iI}^1,\quad \tB^2=-\frac{1}{2}g^{\mu\nu}\partial_\mu\partial_\nu \tG_{\iI}+\tF^2,
\end{equation}
where $\tF^1_{\iI}$ and $\tF^2$ are of order $1$. Generalized hyperbolicity follows from two facts which ensure that the system decouples:
\begin{enumerate}
\item The dependence of $[\tB^1_{\iI}]^{[n]}$ on $D^m\tG_{\jI}^{[k]}$ for $k>n$ is by
\begin{equation}
\left[\left(\frac{d}{2}-1-\rho\partial_\infty\right)\partial_\infty \tG_{\iI}\right]^{[n]}=
\left(\frac{d}{2}-1-n\right)(n+1)\left[\tG_{\iI}\right]^{[n+1]},
\end{equation}
and it does not depend on $D^m\tS^{[k]}$ for $k\geq n$.
\item The dependence of $[\tB^2]^{[n]}$ on $D^m\tS^{[k]}$ for $k>n$ is by
\begin{equation}
\left[\left(\frac{d}{2}-2-\rho\partial_\infty\right)\partial_\infty \tS\right]^{[n]}=\left(\frac{d}{2}-2-n\right)(n+1)\left[\tS\right]^{[n+1]},
\end{equation}
and it does not depend on $D^m\tG_{\iI}^{[k]}$ for $k>n$.
\end{enumerate}
Let us now prove the second statement of the lemma by induction on $k_0$. Suppose that for all $0\leq k<k_0$
\begin{equation}\label{eq:DB-ind}
D^{d-2k-2}\tG_{\iI}^{[k]}=0,\quad D^{d-2k-3} \tS^{[k]}(x)=0,
\end{equation}
then taking up to $d-2k_0-2$ derivatives of $[\tB_{\iI}^1]^{[k_0-1]}=0$ and up to $d-2k_0-3$ derivatives of $[\tB^2]^{[k_0-1]}=0$ we get due to \eqref{eq:DB-ind}
\begin{equation}
D^{d-2k_0-2}\tG_{\iI}^{[k_0]}=0,\quad D^{d-2k_0-3} \tS^{[k_0]}(x)=0,
\end{equation}
which shows the induction together with a trivial statement for $k_0=1$.
\end{proof}

As a result, we obtain:

\begin{lm}\label{lm:gauge-prop}
Suppose that $\tE_{\iI\jI}=O(\rho^{d/2})$ and on the initial surface
\begin{equation}
\tG_{\iI}|_\Sigma=\partial_1\tG_{\iI}|_\Sigma=O(\rho^{d/2}),\ 
\tS|_\Sigma=\partial_1\tS|_\Sigma=O(\rho^{d/2-1}),
\end{equation}
then $\tG_{\iI}=O(\rho^{d/2})$ and $\tS=O(\rho^{d/2-1})$.
\end{lm}

\begin{proof}
It follows from Proposition \ref{lm:B-follows} that $\tB^1_{\iI}=O(\rho^{d/2})$ and $\tB^2=O(\rho^{d/2-1})$ and by generalized hyperbolicity the solution is unique. From linearity, it is just zero.
\end{proof}

Additionally, we have 
\begin{equation}
\tS^{[d/2-1]}=-\frac{1}{2(d/2-1)}g^{[0]\iI\jI}\g^{[d/2]}_{\iI\jI}+\ldots,
\end{equation}
where $\ldots$ denote terms depending only on $\g^{[n]}_{\iI\jI}$ for $n\leq d/2-1$. By modifying $\g^{[d/2]}_{\iI\jI}$ we can assume that $\tS^{[d/2-1]}=0$.

\subsubsection{Gauge fixing conditions}

We assume that on the initial surface
\begin{equation}\label{cond:G0}
D^{d-2} \tG_{\iI}^{[0]}|_\Sigma=0,\quad D^{d-3} \tS^{[0]}|_\Sigma=0,
\end{equation}
and prove that in such a case
\begin{equation}\label{cond:G1}
\tG_{\iI}|_\Sigma=\partial_1\tG_{\iI}|_\Sigma=O(\rho^{d/2}),\ 
\tS|_\Sigma=\partial_1\tS|_\Sigma=O(\rho^{d/2-1}).
\end{equation}
We show it by noticing that the equations $\tB^1_{\iI}=O(\rho^{d/2})$, $\tB^2=O(\rho^{d/2-1})$ hold. We can invoke Lemma \ref{lm:B-closed} to show that
\begin{equation}
D^{d-2k-2}\tG_{\iI}^{[k]}|_\Sigma=0,\quad D^{d-2k-3} \tS^{[k]}|_\Sigma=0,
\end{equation}
Comparing with \eqref{cond:G1} we see that
the missing condition is $\partial_1\tG_{\iI}^{[d/2-1]}|_\Sigma=0$. 

\begin{lm}\label{lm:H-gauge}
Suppose that $\tE_{\iI\jI}=O(\rho^{d/2})$ and on the Cauchy surface $\Sigma$
\begin{equation}
H(\vec{N},\cdot)|_\Sigma=0,\quad \partial_1^k\tG_{\iI}^{[n]}|_\Sigma=0,\ k+2n\leq d-2,\quad  \partial_1^k\tS^{[n]}|_\Sigma=0,\ k+2n\leq d-2,
\end{equation}
then $\partial_1\tG_{\iI}^{[d/2-1]}|_\Sigma=0$ and $D^{d-2k-1} \g^{[k]}_{\iI\jI}|_\Sigma=D^{d-2k-1} H^k_{\RS,\ \iI\jI}(\g^{[0]}_{\kI\lI})|_\Sigma$.
\end{lm}

\begin{proof}
We will assume for simplicity that $\vec{N}=\partial_1$. The modification of the proof for the general case is minor.
 We have the identity for jets
\begin{equation}
D^{d-2n-3}\RS_{\iI\jI}^{[n]}|_\Sigma=D^{d-2n-3}\left[\tE_{\iI\jI}+\frac{1}{2}(\tnabla_{\iI}\tG_{\jI}+\tnabla_{\jI}\tG_{\iI})+\g_{\iI\jI}\tS\right]^{[n]}|_\Sigma=0
\end{equation}
From Lemma \ref{lm:recursive} the jets of the expansion of the ambient metric
\begin{equation}
D^{d-2k-1} \g^{[k]}_{\iI\jI}|_\Sigma,
\end{equation}
agree with the Fefferman-Graham ambient metric construction. As $\RS_{1\jI}^{[d/2-1]}$ depends only on first derivatives of $\g_{\iI\jI}^{[d/2-1]}$ in $x^1$ direction, we deduce
\begin{equation}
0=H_{1\jI}(h)|_\Sigma=\RS_{1\jI}^{[d/2-1]}|_\Sigma-\frac{1}{d/2-1}\g_{1\jI}^{[0]}\RS_{\infty\infty}^{[d/2-2]}|_\Sigma=\RS_{1\jI}^{[d/2-1]}|_\Sigma,
\end{equation}
by using \eqref{eq:obstruction} and
$\RS_{\infty\infty}^{[d/2-2]}|_\Sigma=0$ (due to $\tS^{[d/2-1]}|_\Sigma=0$).
Additionally, from $\tS|_\Sigma=O(\rho^{d/2})$ and $\tG_{\iI}|_\Sigma=O(\rho^{d/2})$ it follows that
\begin{equation}
0=\RS_{1\jI}^{[d/2-1]}|_\Sigma=\tE_{1\jI}^{[d/2-1]}|_\Sigma+\frac{1}{2}([\tnabla_{1}\tG_{\jI}]^{[d/2-1]}|_\Sigma+[\tnabla_{\jI}\tG_{1}]^{[d/2-1]}|_\Sigma)+[\g_{1\jI}\tS]|_\Sigma^{[d/2-1]}=
\frac{1}{2}\partial_{1}\tG_{\jI}^{[d/2-1]}|_\Sigma,
\end{equation}
for $\jI\not=1$. Similarly, for $\jI=1$ we get
\begin{equation}
0=\RS_{11}^{[d/2-1]}|_\Sigma=\tE_{11}^{[d/2-1]}|_\Sigma+\frac{1}{2}([\tnabla_{1}\tG_{1}]^{[d/2-1]}|_\Sigma+[\tnabla_{1}\tG_{1}]^{[d/2-1]}|_\Sigma)+[\g_{11}\tS]|_\Sigma^{[d/2-1]}=
\partial_{1}\tG_{1}^{[d/2-1]}|_\Sigma,
\end{equation}
showing the statement of the lemma.
\end{proof}

We can summarize the results obtained so far in the following proposition:

\begin{prop}\label{prop:AFG}
We consider initial data $\partial_1^kh_{\iI\jI}|_\Sigma$, $k=0,\ldots d-1$.
Suppose that on the initial surface $\Sigma$
\begin{equation}
D^{d-2} (\square x^{\jI})|_\Sigma=0,\quad D^{d-3} (R)|_\Sigma=0,\ H(\vec{N},\cdot)|_\Sigma=0
\end{equation}
The conditions are well-defined because $\square x^{\iI}$ depends on on up to first jets of the metric, $R$ depends on up to the second jets and $H(\vec{N},\cdot)$ depends on $d-1$ jets of the metric on $\Sigma$.
Then there exists a unique solution to AFG system $H_{\iI\jI}=0$ with the given initial data and which satisfies $\square x^{\jI}=0$, $R=0$.
\end{prop}

\begin{proof}
As $\tE_{\iI\jI}$ is recursive till order $d/2-1$ we compute by Lemma \ref{lm:recursive}  the initial value data for the system by
\begin{equation}
\tilde{h}_{\iI\jI}^{[k]}|_\Sigma=H_{\tE\ \iI\jI}^{k}(D^{2k} h_{\iI\jI}|_\Sigma),\quad \partial_1\tilde{h}_{\iI\jI}^{[k]}|_\Sigma=\partial_1H_{\tE\ \iI\jI}^{k}(D^{2k} h_{\iI\jI}|_\Sigma),\quad k\leq d/2-1.
\end{equation}
This allows us to determine the initial data for the equation $\tE=O(\rho^{d/2})$. 
We consider a unique solution $\g_{\iI\jI}$ of this equation. 
Since the system is generalized hyperbolic, we get
\begin{equation}
\partial_1^n \g^{[0]}_{\iI\jI}|_\Sigma=\partial_1^n h_{\iI\jI}|_\Sigma,\quad n\leq d-1,
\end{equation}
 by Lemma \ref{lm:bijective}.
Thus, the solution has prescribed initial data. 

In particular, it is true that
\begin{equation}
D^{d-2} \tG_{\iI}^{[0]}|_\Sigma=D^{d-2}\left(h_{\iI\jI}h^{\kI\lI}\Gamma^{\jI}_{\kI\lI}\right)=-D^{d-2}(h_{\iI\jI}\square x^{\jI})=0.
\end{equation}
and
\begin{equation}
D^{d-3}\left[R_{\iI\jI}-\left(\frac{d}{2}-1\right)\g_{\iI\jI}^{[1]}\right]|_\Sigma=D^{d-3}\left[\tE_{\iI\jI}+\frac{1}{2}(\tnabla_{\iI}\tG_{\jI}+\tnabla_{\jI}\tG_{\iI})\right]^{[0]}|_\Sigma=0.
\end{equation}
We take the trace of this equality to derive
\begin{equation}
D^{d-3}\tS^{[0]}|_\Sigma=D^{d-3}\left[-\frac{1}{2}\g^{[0]\kI\lI}\g_{\kI\lI}^{[1]}\right]|_\Sigma=-D^{d-3}\frac{R}{d-2}|_\Sigma=0,
\end{equation}
This means that the condition \eqref{cond:G0} is satisfied.
From Lemma \ref{lm:gauge-prop} and Lemma \ref{lm:H-gauge} we conclude
\begin{equation}
\tG_{\iI}=O(\rho^{d/2}),\ \tS=O(\rho^{d/2-1}).
\end{equation}
We can always assume $\tS^{[d/2-1]}=0$. Taking this into account, we obtain
\begin{equation}
\RS_{\iI\jI}^{[n]}=\tE_{\iI\jI}^{[n]}+\frac{1}{2}([\tnabla_{\iI}\tG_{\jI}]^{[n]}+[\tnabla_{\jI}\tG_{\iI}]^{[n]})+[\g_{\iI\jI}\tS]^{[n]}=0,
\end{equation}
for $n\leq d/2-1$.
We have a solution with
$\RS_{\iI\jI}=O(\rho^{d/2})$.
\end{proof}

\subsection{The AFG equation in the Anderson-Chru{\'s}ciel gauge}

In this section, the correspondence of our gauge functions to $R=0$ and $\square x^\mu=0$ gauge will be investigated. This gauge was proposed in \cite{Anderson} and \cite{Anderson-Chrusciel}.

\subsubsection{Uniqueness of the solution of the AFG equation}

Assume that we have a solution $H_{\iI\jI}(h)=0$ with the given initial conditions at $\Sigma$. 

\begin{lm}[\cite{Anderson}, \cite{Anderson-Chrusciel}]
Suppose that we have a local coordinate system $y^2,\ldots y^d$ on $\Sigma$.
Locally there exists a coordinate system $x^{\iI}$ and the choice of the conformal factor $f$ such that for $h'_{\iI\jI}=e^{2f}h_{\iI\jI}$
\begin{enumerate}
\item $\square'x^{\iI}=0$, $R'=0$,
\item $x^{\kI}|_\Sigma=y^{\kI}$ for $\kI=2,\ldots d$, $x^1|_\Sigma=0$,
\item $\partial_1$ is a unit normal vector to $\Sigma$
\end{enumerate}
Here $\square'$ is a scalar d'Alembert operator with respect to the metric $h'$.
\end{lm}

\begin{proof}
First we find $f$ as a solution of Yamabe problem, which is a nonlinear hyperbolic system for $f$ (see \cite{Anderson-Chrusciel} for discussion). Define $x^{\kI}$ for $\kI\geq 2$ as a unique solution to $\square' \phi=0$ with initial conditions (local development)
\begin{equation}
\phi|_U=y^{\kI},\quad N^{\iI}\partial_{\iI}\phi|_U=0,
\end{equation}
where $N^{\iI}$ is a unit normal to $\Sigma$. Finally, we define $x^1$ as a unique solution with initial value
\begin{equation}
\phi|_U=0,\quad N^{\iI}\partial_{\iI}\phi|_U=1.
\end{equation}
One can check that at $U'\subset U$ this new coordinates are independent, so it is also true in the small neighborhood of the Cauchy surface.
\end{proof}

We can thus work in this gauge.
From the solution to $H_{\mu\nu}=0$ we now construct iteratively 
\begin{equation}
\g_{\iI\jI}^{[n+1]}\quad n\leq d/2-2.
\end{equation}
by \eqref{eq:g-S-iter}.
Let us notice that $\g^{[0]}_{\iI\jI}=h_{\iI\jI}$ and $\g^{[1]}_{\iI\jI}=P_{\iI\jI}$ where $P_{\iI\jI}=\frac{1}{d-2}\left(R_{\iI\jI}-\frac{1}{2(d-1)}Rh_{\iI\jI}\right)$ is the Schouten tensor \cite{Fefferman-Graham}. Due to the gauge condition, one obtains
\begin{equation}
\tG^{[0]}_{\iI}=h_{\iI\jI}h^{\kI\lI}\Gamma^{\jI}_{\kI\lI}=-h_{\iI\jI}\square x^{\jI}=0,
,\quad \tS^{[0]}=\frac{1}{2}h^{\iI\jI}P_{\iI\jI}=\frac{1}{2(d-1)}R=0.
\end{equation}
Moreover, from $\RS_{\iI\infty}=O(\rho^{d/2-1})$ and $\RS_{\infty\infty}=O(\rho^{d/2-1})$ we have
\begin{equation}
\tG_{\iI}=O(\rho^{d/2}),\quad \tS=O(\rho^{d/2}).
\end{equation}
Additionally,  $\RS_{\iI\jI}=O(\rho^{d/2})$ and so
\begin{equation}\label{eq:E-again}
\tE_{\iI\jI}=O(\rho^{d/2}).
\end{equation}
From the uniqueness of the solution of \eqref{eq:E-again} we obtain the uniqueness of the solution of AFG equation (in the given gauge).

\subsubsection{Existence of the solutions of the AFG equation}

Let us now assume that the initial data is given by
\begin{equation}
\partial_1^nh_{\iI\jI}|_\Sigma,\quad n=0,\ldots d-1,
\end{equation}
which satisfies the constraints $H(\vec{N},\cdot)|_\Sigma=0$ and the
gauge is satisfied:
\begin{align}\label{eq:x-1}
\partial_1^n\square x^{\iI}|_\Sigma=0,\quad n=0,\ldots d-2,\\
\partial_1^n R|_\Sigma=0,\quad n=0,\ldots d-3.\label{eq:x-2}
\end{align}
The conditions are well-defined because $\square x^{\iI}$ depends on up to first jets of the metric and $R$ depends on up to the second jet.

By a change of coordinates in jets of $\Sigma$ we can always assume these conditions together with $\vec{N}=\partial_1$ where $N$ is a normal vector. In this way, the constraints take the form
\begin{equation}
H_{1\iI}|_\Sigma=0.
\end{equation}
We also compute
\begin{equation}
D^{d-2} \tG_{\iI}^{[0]}|_\Sigma=-D^{d-2}(h_{\iI\jI}\square x^{\jI})|_\Sigma=0,\quad D^{d-3} \tS^{[0]}|_\Sigma=\frac{1}{2(d-2)}D^{d-3}R|_\Sigma=0,
\end{equation}
by \eqref{eq:x-1} and \eqref{eq:x-2}. The existence of the solution of AFG equation would follow from Proposition \ref{prop:AFG} if we were able to use this gauge globally on a compact $\Sigma$. However, it is not possible, so we need to apply some version of a gluing argument.

\subsubsection{Proof of the Theorem \ref{thm:AFG}}

The harmonic gauge is well-suited for $\Sigma=\R^{d-1}$. If we want to apply our result, we need to extend the notion of this gauge to compact Cauchy surfaces. This can be done in the case of $\Sigma$ being a torus, where $x^\mu$ for $\mu=2,\ldots d$ are defined modulo $2\pi$. 
Due to the finite speed of propagation, the method provides an existence and uniqueness result also for open subsets of the torus. Uniqueness of the development allows to apply the standard gluing argument \cite{Choquet-Bruhat} to obtain Theorem \ref{thm:AFG}.

\subsection{Infinite order extension of the ambient metric}

Suppose that the obstruction tensor vanishes.
The results of \cite{Fefferman-Graham} show that Taylor expansions of the metrics, which are Ricci flat of the order $O(\rho^\infty)$, are in one-to-one correspondence with the traceless symmetric tensors $k_{\mu\nu}$ satisfying
\begin{equation}\label{eq:k}
\nabla^\mu k_{\mu\nu}=D_\nu,
\end{equation}
where $D_\nu$ is a certain $1$-form (defined in eq. 3.36 in \cite{Fefferman-Graham}). The tensors $k_{\mu\nu}$ define trace-free part of $\g^{[d/2]}_{\iI\jI}$ since the trace is already determined. In the case of Euclidean manifolds, it is not obvious that such a tensor exists. We will prove that this is the case for any globally hyperbolic spacetime.

\begin{prop}
There exists $k_{\mu\nu}$ satisfying \eqref{eq:k} on a AFG globally hyperbolic spacetime.
\end{prop}

\begin{proof}
We will look for the tensor given in a special form
\begin{equation}
k_{\mu\nu}=\nabla_\mu u_\nu+\nabla_\nu u_\mu-\frac{2}{d}g_{\mu\nu} \nabla^\rho u_\rho,
\end{equation}
for some covector field $u_\mu$. It is already symmetric, traceless and the equation \eqref{eq:k} takes a form
\begin{equation}\label{eq:d}
D_\nu=\nabla^\mu\nabla_\mu u_\nu+R^\rho_\nu u_\rho+\left(1-\frac{2}{d}\right)\nabla_\nu(\nabla^\rho u_\rho).
\end{equation}
Taking the divergence, we obtain an additional equation
\begin{equation}\label{eq:nablad}
\nabla^\nu D_\nu=\nabla^\mu\nabla_\mu (\nabla^\nu u_\nu)+R^{\mu\nu}(\nabla_\mu u_\nu-\nabla_\nu u_\mu)+2\nabla^\mu(R^\nu_\mu u_\nu)+\left(1-\frac{2}{d}\right)\nabla^\mu\nabla_\mu(\nabla^\nu u_\nu),
\end{equation}
where we used $\nabla^\nu\nabla^\mu\nabla_\mu  u_\nu=\nabla^\mu\nabla_\mu (\nabla^\nu u_\nu)+\nabla^\mu(R^\nu_\mu u_\nu)$. We now introduce a new variable $A=\frac{1}{d}\nabla^\mu u_\mu$ and a system of equations (equivalent to \eqref{eq:nablad} and \eqref{eq:d})
\begin{align}
&\nabla^\mu\nabla_\mu u_\nu+R^\rho_\nu u_\rho+\left(d-2\right)\nabla_\nu A=D_\nu,\label{eq:k1}\\
&2(d-1)\nabla^\mu\nabla_\mu A+2\nabla^\mu(R^\nu_\mu u_\nu)=\nabla^\nu D_\nu.\label{eq:k2}
\end{align}
It is indeed a hyperbolic second-order linear system. Thus, with the given initial data on a Cauchy surface, it has a solution. We now notice that the divergence of the left hand side of \eqref{eq:k1} minus left hand side of \eqref{eq:k2} is equal to zero:
\begin{equation}
\nabla^\mu\nabla_\mu (\nabla^\rho u_\rho-dA)=0.
\end{equation}
If we choose $A|_\Sigma=d^{-1}\nabla^\rho u_\rho|_\Sigma$ and $\partial_1 A|_\Sigma=d^{-1}\partial_1\nabla^\rho u_\rho|_\Sigma$ (computed by Cauchy-Kovalevskaya algorithm), then $A=\frac{1}{d}\nabla^\rho u_\rho$ in the whole spacetime and $\nabla^\mu k_{\mu\nu}=D_\nu$.
\end{proof}

We can thus always assume that the metric is Ricci flat to an infinite order, but it is not uniquely defined except terms  $\g^{[n]}_{\iI\jI}$ for $n\leq d/2-1$ and $\operatorname{tr} \g^{[d/2]}$.

\section{GJMS type operators for tractor bundles}

We will now concentrate on various linear systems, which arise by the ambient metric construction. They share the common property with $\tE_{\iI\jI}$, that the principal symbol is a power of the d'Alembert operator. In this part of the paper, we will also shortly describe Graham-Jenne-Mason-Sparling (GJMS) type systems. A quite general method of introducing this type of operators is by tractor calculus. We will only focus on the essential parts of this theory in terms of the ambient metric construction (see \cite{Cap2002}). For the short review of the tractor calculus in application to general relativity, we refer reader to \cite{Curry2015}.

We work on manifold $\mM=\R\times \tilde{M}$, $\mT=t\partial_t$ is a conformal Killing vector with property that for every vector field~$\mF^I$ it satisfies
\begin{equation}
\mnabla_{\mF}\mT=\mF, \quad \mnabla_I\mT_J=\mg_{IJ}.
\end{equation}
We also introduce $\mr=\frac{1}{2}\mT^I\mT_I=\rho t^2$ with the properties:
\begin{equation}
\mnabla_I\mr=\mT^J\mnabla_I\mT_J=\mT_I,\quad \mnabla^I\mnabla_I \mr=\mnabla^I\mT_I=(d+2),\quad \mnabla_I\mr\mnabla^I\mr=2\mr.
\end{equation}
Simplifying the notation, we will often skip tractor indices. We consider $n$-covector $\mX_{I_1\ldots I_n}$ with the property
\begin{equation}
\Lie_{\mT} \mX=w\mX,
\end{equation}
for $w\in \R$. We call $w$ a weight of the field $\mX$. 
The field $\mX$ is determined by its restriction to $t=1$
\begin{equation}
\tX=\mX|_{t=1},
\end{equation}
According to \cite{Cap2002},  $\tX^{[0]}$ is a section of a tractor bundle $\epsilon_{I_1\ldots I_n}[w-n]$. The ambient space (at least the uniquely determined jets of the metric) is a natural construction for the conformal structure on the manifold. The natural bundles inherit the laws of transformations under diffeomorphisms and conformal transformations of the original metric. We will not describe here the original formulation (see for example, \cite{Cap2002,Curry2015}), but we will use this description as a definition.

We will need the following results:

\begin{lm}
Suppose that $\Lie_{\mT} \mX_{I_1\ldots I_n}=w\mX_{I_1\ldots I_n}$ then
$\mnabla_{\mT}\mX=(w-n)\mX$, where $n$ is a valency of the field.
\end{lm}

\begin{proof}
We use induction on $n$. For $n=0$ (that means $\mX$ is a functions) Lie derivative and covariant derivative agrees. Suppose that the result is true for $n<n_0$. Consider a vector field $\mF^I$ with weight $0$. We have
\begin{equation}
\mnabla_{\mT}\mF=\Lie_{\mT}\mF+\mnabla_{\mF}\mT=\mF.
\end{equation}
For any $\mX_{I_1\ldots I_{n_0}}$, $n_0$ covector with weight $w$, we denote by
$\mF\llcorner\mX$ 
\begin{equation}
\mX_{I_1\ldots I_{n_0}}\mF^{I_1}.
\end{equation}
It is a $n_0-1$ covector with weight $w$, thus
\begin{equation}
(w-(n_0-1))\mF\llcorner\mX=\mnabla_{\mT}(\mF\llcorner\mX)=\mF\llcorner(\mnabla_{\mT}\mX)+(\mnabla_{\mT}\mF)\llcorner\mX=\mF\llcorner(\mnabla_{\mT}\mX)+\mF\llcorner\mX.
\end{equation}
So $\mF\llcorner[(w-n_0)\mX-\mnabla_{\mT}\mX]=0$. As the restriction to $t=1$, $\tF^I$ is arbitrary, we show the induction.
\end{proof}

The GJMS operators are constructed with the help of the d'Alembert operator in the ambient space $\mM$.
We consider an operator $\msquare$ defined on $n$-covectors:
\begin{equation}
\msquare \mX=\mnabla^I\mnabla_I \mX.
\end{equation}
The result has the weight $w-2$. We can thus define
\begin{equation}
\tboxdot_w \tX=[\msquare \mX]|_{t=1}\text{ for } \tX=\mX|_{t=1},\ \Lie_{\mT}\mX=w\mX.
\end{equation}

\begin{prop}\label{prop:form-GJMS}
For any $n$, the operator $\tboxdot_w$ on a $\rho$ dependent $n$-tractors of weight $w$  $\tX$ has the property
\begin{equation}
\tboxdot_w \tX=\g^{\mu\nu}\partial_\mu\partial_\nu \tX+\left(d-2+2w-2n-2\rho\partial_\infty\right)\partial_\infty\tX+\tF,
\end{equation}
where $\tF^{[m]}$ depends only on $D^1\tX^{[k]}$ for $k\leq m$.
\end{prop}

In other words
\begin{equation}
[\tboxdot_w \tX]^{[m]}=[\g^{\mu\nu}\partial_\mu\partial_\nu \tX]^{[m]}+\left(d-2+2w-2n-2m\right)(m+1)\tX^{[m+1]}+\tF^{[m]}.
\end{equation}

\begin{proof}
Clearly, we can write
\begin{equation}\label{eq:exp-derivatives}
\tboxdot_w \tX=\g^{\mu\nu}\partial_\mu\partial_\nu \tX+\tilde{G},
\end{equation}
where $\tilde{G}^{[m]}$ depends on $D^1\tX^{[k]}$ for $k\leq m+1$ and $\tX^{[m+2]}$.

We need to determine dependence of $[\tboxdot \tX]^{[m]}$ on $D^l\tX^{[k]}$ for $k>m$. This can be done by considering fields, which depend only on the Taylor expansion coefficients for $k>m$. The most convenient way is to use $\mr=\rho t^2$ in the expansion instead of $\rho$, because $\mr$ is covariantly defined. We consider $\mX=\mr^{m+1}\mF=O(\rho^{m+1})$ where $\Lie_{\mT}\mF=(w-2(m+1))\mF$ (in order that $\mX$ has the proper weight) and thus
\begin{equation}
\mnabla_{\mT}\mF=(w-2(m+1)-n)\mF.
\end{equation}
We substitute the special form of $\mX$ in the following formula
\begin{equation}
\mnabla^I\mnabla_I \mX=(m+1)\mr^{m-1}\left[\mF\left(m\mnabla^I\mr \mnabla_I\mr+\mr(\mnabla^I\mnabla_I\mr)\right)+
 2\mr\nabla^I\mr \nabla_I\mF\right]+\mr^{m+1}\mnabla^I\mnabla_I\mF.
\end{equation}
We use  the know form of derivatives of $\mr$ to get the nice expression:
\begin{align}
&\mnabla^I\mnabla_I \mX=(m+1)\mr^m\left((2m+d+2)\mF+2\mnabla_T\mF\right)+O(\rho^{m+1})=\nonumber\\
&=(m+1)\mr^m\mF\left(d-2+2w-2n-2m\right)+O(\rho^{m+1}).
\end{align}
This shows that 
\begin{equation}
\tboxdot_w \tX=\left(d-2+2w-2n-2\rho\partial_\infty\right)\partial_\infty\tX+\tilde{H}
\end{equation}
where $\tilde{H}^{[m]}$ depends on $D^2\tX^{[k]}$ for $k\leq m$. Together with the previous expansion \eqref{eq:exp-derivatives} it proves the lemma.
\end{proof}

Together with Proposition \ref{prop:hip} this result leads immediately to a corollary:

\begin{cor}\label{cor:GJMS}
Let $n\in \Z_+\cup\{0\}$ and  $w\in\Z$ such that $N=d/2-1+w-n\in \Z_+\cup\{0\}$, then the system 
\begin{equation}
\tboxdot_w\tX+\tD(\tX)=O(\rho^{N+1})\text{ for } \tX^{[k]},\ k=0,\ldots, N
\end{equation}
for $\tX^{[k]}$ ($n$-tractors of weight $w$)
is generalized hyperbolic, recursive and linear. Here $\tD$ is linear transformation on the space of $n$-tractors. In particular, the Cauchy problem with smooth initial data is well-posed.
\end{cor}

\begin{proof}
The assumption about the metric is sufficient to satisfy the requirements for the linear generalized hyperbolic system to be well-posed in Proposition \ref{prop:hip}.
\end{proof}

\begin{remark}\label{remark:GJMS}
The ambient metric construction determines $\g^{[n]}_{\iI\jI}$ for $n=0,\ldots, d/2-1$ and $\operatorname{tr}\g^{[d/2]}$. Consequently, only equations depending on this part of the metric expansion are defined uniquely. They provide conformal equations on the spacetime $M$. 
This holds if $0>w-n>1-d/2$. Interestingly, it is also true for $\tboxdot_0$ on scalars of weight $0$, where the operator depends on the aforementioned trace. If $n\geq 1$ and $w-n\geq 0$, then the system of equations explicitly depends on the choice of the extension of the ambient metric.
\end{remark}

We will denote the derived equation for $\tsquare_w$ by $G_w$.
The initial data for the derived equation is given by $\partial_1^k\tX^{[0]}|_\Sigma$, $k\leq 2N+1$ and we get from Lemma \ref{lm:bijective} and Proposition \ref{prop:hip} the unique global development. 

\subsection{Spacetimes with the  vanishing \texorpdfstring{$Q$}{Q} curvature}\label{sec:Q}

An immediate application of Collorary \ref{cor:GJMS} is uniqueness and existence of the Cauchy development for GJMS operators as they are the derived system for $\tboxdot_w$ acting on scalar functions of weight $-k$, $0\leq k\leq d/2-1$. We can apply the result to the problem of finding a conformal factor, which yields the vanishing Branson $Q$ curvature \cite{Branson}. This is an important and quite mysterious object in conformal geometry (see \cite{Chang2008, Baum2010} for an introduction to the application and meaning of the $Q$ curvature).

\begin{prop}
On every globally hyperbolic spacetime, there exists a function $\phi$ such that $Q(e^{2\phi}h)=0$.
\end{prop}

\begin{proof}
For the metric  $h$ the Branson $Q$ curvature can be constructed as follows. We find a scalar function $\mf$ such that 
\begin{equation}
\Lie_{\mT}\mf=1,\quad \mf^{[0]}|_{t=1}=0,\quad \msquare \mf=O(\rho^{d/2-1}).
\end{equation}
It is always possible to find such a function and it is unique up to $O(\rho^{d/2})$. We define $Q=c_d[\msquare\mf]|_{t=1}^{[d/2-1]}$ where $c_d$ is a dimension dependent constant \cite{Fefferman2003}. Let us notice that
\begin{equation}
\Lie_{\mT}\mF=0,\quad \mF=\mf-\ln t.
\end{equation}
Thus, it is a scalar function of weight $0$.

The conformal rescaling of the metric corresponds to a diffeomorphism, which acts by $t'=te^{-\phi}$ at $\rho=0$. Suppose that $\mf$ satisfying
\begin{equation}
\Lie_{\mT}\mf=1,\quad \msquare \mf=O(\rho^{d/2}),
\end{equation}
then we can define $t'=e^{\mf^{[0]}}=te^{\mF^{[0]}}$. The metric rescaled by $e^{2\mF^{[0]}}$ has a vanishing Branson curvature.

Let us now notice that the equation
\begin{equation}
\msquare\mF=\msquare\ln t+O(\rho^{d/2})
\end{equation}
is well-posed for $\tF=\mF|_{t=1}$. For linear (affine) systems, there exists a solution on the whole globally hyperbolic spacetime, if we provide arbitrary initial data on some Cauchy surface. 
\end{proof}

For future reference, we can also compute
\begin{equation}\label{eq:lnt}
\msquare\ln t=\frac{1}{\sqrt{\mg}}\partial_I\left(\sqrt{\mg}\mg^{I0}t^{-1}\right)=t^{-2}\frac{1}{\sqrt{\g}}\partial_\infty \sqrt{\g}.
\end{equation}
We notice that up to terms of order $O(\rho^{d/2})$ it depends only on the  part of the metric, which is determined by the ambient construction.

\subsection{Propagation of the the Einsteinian condition}

Let us introduce a tractor connection (see \cite{Curry2015}). Let $X_{I_1\cdots I_n}$ be an  $n$-tractor of weight $w$ field on the spacetime $M$. We define for a vector field $Y$
\begin{equation}
\nabla^{\mathcal T}_Y X_{I_1\cdots I_n}=\left[\mnabla_\mY \mX_{I_1\cdots I_n}\right]_{t=1,\rho=0},
\end{equation}
where we choose $\mX$ such that $\Lie_{\mT}\mX=w\mX$ and $\mX|_{t=1,\rho=0}=X$ and $\mY^\mu=Y^\mu$, $\mY^0=\mY^\infty=0$. The definition does not depend on a particular choice of $\mX$. If $w=n$ then the tractor derivative has particularly nice properties. Usually, one restricts the definition to this special type of tractor. We define the tractor derivative for $1$-tractor $X_I$ of weight $1$ (see tractor derivative, for example, in \cite{Graham2011} with a correction for covectors) by the formula:
\begin{equation}\label{eq:tractor-connection}
\nabla^{\mathcal T}_{\iI} \left(\begin{array}{c}
X_0\\ X_{\jI}\\ X_\infty\end{array}\right)=\left(\begin{array}{c}
\nabla_{\iI}X_0-X_{\iI}\\ \nabla_{\iI}X_{\jI}+g_{\iI\jI}X_\infty+P_{\iI\jI}X_0\\ \nabla_{\iI}X_\infty-P_{\iI}^{\jI}X_{\jI}\end{array}\right).
\end{equation}
An almost conformally Einstein metric (according to Gover \cite{Gover2008}) is a metric for which there exists $X_I$ of weight $1$ satisfying
\begin{equation}\label{eq:tractor-const}
\nabla^{\mathcal T}_\mu X_I=0.
\end{equation}
This is equivalent to
\begin{equation}
X_0=f,\ X_\mu=\partial_\mu f,\ X_\infty=-\frac{1}{d}(\nabla^\mu\nabla_\mu f+P^\mu_\mu f),
\end{equation}
for a function $f$ satisfying
\begin{equation}
\tf (\nabla_\mu\nabla_\nu f- P_{\mu\nu}f)=0,
\end{equation}
where $\tf$ denotes  a trace-free part (see, for example \cite{Curry2015}).
It is known \cite{Gover2008} that the metric is conformal to the Einstein metric if and only if  it is almost conformally Einstein and the following non-degeneracy condition holds
\begin{equation}
X_0\not=0.
\end{equation}
The conformal factor rescaling metric to Einsteinian is given by $e^\phi=X_0$ and the cosmological constant is equal $\Lambda=c_dX^IX_I$  where $c_d$ is a dimension dependent constant \cite{Gover2004a}. Conformal boundary corresponds to $X_0=0$. This set is a hypersurface with vanishing extrinsic curvature and it enjoys very special properties \cite{Gover2004a}. In almost conformally Einstein spaces, the Fefferman-Graham obstruction tensor vanishes.

It is known \cite{Graham2011} that one can prolong the covariant tractor from $\rho=0$ surface. We will need a detailed statement of this result. We remind  the following fact (proven in \cite{Graham2011}, see also \cite{Gover2008})

\begin{prop}\label{prop:Gover-Graham}
Suppose that $\nabla^{\mathcal T}_\mu X_I=0$ for a tractor of weight $1$, then there exists $\mX_I$ of weight $1$ on the ambient space such that $\mX_I|_{t=1,\rho=0}=X_I$ and
\begin{equation}
\mnabla_{I}\mX_J=O(\rho^{d/2-1}),
\end{equation}
\end{prop}

In particular, in the case of an Einstein metric $h_{\iI\jI}$, we define an infinite order extension
\begin{equation}\label{eq:g0}
\g^0_{\iI\jI}=(1+\lambda \rho)^2 h_{\iI\jI},
\end{equation}
where $\lambda=\frac{R}{2d(d-1)}$ which is Ricci flat.
Then $\mX_I=\partial_I(t(1-\lambda\rho))$ satisfies
\begin{equation}
\mnabla_{I}\mX_J=O(\rho^{\infty}),
\end{equation}
We introduce $\mf=\mT^J\mX_J=t(1-\lambda\rho)$. Let us notice that
$\mnabla_I (\mT^J\mX_J)=\mX_I+\mT^J\mnabla_I \mX_J=\mX_I+O(\rho^{\infty})$ and then $\mnabla_I\mnabla_J \mf=O(\rho^{\infty})$. In particular, we see that $\msquare\mf=O(\rho^{\infty})$.

In general, we do not have a distinguished infinite order extension unless we already know the covariant tractor. We would like to determine how many of the properties of $\mf=\mT^I\mX_I$ extend to an arbitrary Ricci flat extension and to the conformal boundary. Suppose in a first step that we work in an arbitrary Ricci flat extension of a non-degenerate almost Einstein manifold. We compute for $\mf=t(1-\lambda\rho)$
\begin{equation}
\msquare \mf=\left(1+\lambda \rho\right)\partial_\infty\ln \sqrt{\g}-d\lambda.
\end{equation}
Let us introduce $\tilde{A}^{\iI}_{\jI}:=\tilde{g}^{\iI\kI}\g_{\kI\jI}'$. For an arbitrary extension of the Einsteinian ambient metric, we have
\begin{equation}
\tilde{A}^{\iI}_{\jI}=\frac{2\lambda}{1+\lambda\rho}\delta^{\iI}_{\jI}+\tilde{k}^{\iI}_{\jI},\quad \tilde{k}^{\iI}_{\jI}=O(\rho^{d/2-1}).
\end{equation}
The condition $\RS_{\infty\infty}=O(\rho^\infty)$ gives 
\begin{equation}
-\frac{1}{2}\partial_\infty\tilde{A}^{\iI}_{\iI}-\frac{1}{4}\tilde{A}^{\iI}_{\jI}\tilde{A}^{\jI}_{\iI}=O(\rho^\infty),
\end{equation}
by \eqref{eq:S-infty2}. Consequently, we obtain a property of $\tilde{k}^{\iI}_{\jI}$:
\begin{equation}
\frac{1}{2}\partial_\infty\tilde{k}^{\iI}_{\iI}-\frac{\lambda}{1+\lambda\rho}\tilde{k}^{\iI}_{\iI}=O(\rho^{d-1})\Longrightarrow \tilde{k}^{\iI}_{\iI}=O(\rho^{d-1}).
\end{equation}
Finally, we get
\begin{equation}
\partial_\infty \ln\sqrt{\g}=\frac{1}{2}\tilde{A}^{\iI}_{\iI}=\partial_\infty \ln\sqrt{\g^0}+O(\rho^{d-1}),
\end{equation}
where $\g^0_{\iI\jI}$ is given by \eqref{eq:g0}.
Independently of the Ricci flat extension, the following is true
\begin{equation}\label{eq:mf-hip}
\msquare \mf=O(\rho^{d/2+1}).
\end{equation}
We will see that the condition \eqref{eq:mf-hip} can be also satisfied on the conformal boundary. In general, we cannot ensure vanishing of $\msquare \mf$  to higher order on the boundary unless the  metric is even with respect to the conformal boundary (see \cite{Fefferman-Graham} for the definition of even Poincare-Einstein metrics).

Equation \eqref{eq:mf-hip} is a quite important observation. As $\mf$ is a scalar of weight $1$, this is a generalized hyperbolic system for $\tilde{f}=\mf|_{t=1}$. We can evolve it from the Cauchy surface to the whole globally hyperbolic development. In order to define evolution we need to specify the choice of extension of the ambient metric. We choose arbitrary extension that is Ricci flat to all orders. We will now show that if \eqref{eq:mf-hip} holds, then $\mnabla_I\mX_J$ satisfies the linear hyperbolic equation too and then what remains is to show that the initial data for this system vanish.

Let us notice an identity for an arbitrary function $\mF$
\begin{equation}\label{eq:commutator}
\mboxplus(\mnabla_I\mnabla_J\mF)=\mnabla_I\mnabla_J(\msquare \mF)+(\mnabla_J\mRS_{I}^K+\mnabla_I\mRS_{J}^K-\mnabla^K\mRS_{IJ})\mnabla_K\mF,
\end{equation}
where we introduced an operator $\mboxplus$ on $2$- covectors $\mD_{IJ}$ on $\mM$
\begin{equation}
\mboxplus(\mD)_{IJ}=\mnabla^L \mnabla_L \mD_{IJ}+2\mRS^{K\phantom{I}L}_{\phantom{K}I\phantom{L}J}\mD_{KL}-\mRS^{K}_I \mD_{KJ}-\mRS^{K}_J \mD_{KI},
\end{equation}
where $\mRS^{K\phantom{I}L}_{\phantom{K}I\phantom{L}J}$, $\mRS^{K}_I$ are the Riemann tensor and Ricci tensor on $\mM$. If $\mRS_{IJ}=O(\rho^\infty)$ and \eqref{eq:mf-hip} holds, then from  \eqref{eq:commutator} it follows that
\begin{equation}
\mboxplus (\mD_{IJ})=\mnabla_I\mnabla_J(\msquare \mf)+O(\rho^{\infty})=O(\rho^{d/2-1}),\quad \mD_{IJ}=\mnabla_I\mnabla_J\mf.
\end{equation}
This is also a hyperbolic equation for $\tD_{IJ}^{[k]}$, $k=0,\ldots d/2-2$ ($\tD_{IJ}=\mD_{IJ}|_{t=1}$). We thus need to show that
\begin{equation}\label{eq:DD}
\tD_{IJ}|_\Sigma=O(\rho^{d/2-1}),\ \partial_1\tD_{IJ}|_\Sigma=O(\rho^{d/2-1}).
\end{equation}

\begin{prop}\label{prop:tractor}
Suppose that we have initial data $D^{d+1}f|_\Sigma$  such that 
\begin{equation}\label{eq:tractor-cond-D}
D^{d-1} \tf (\nabla_\mu\nabla_\nu f- P_{\mu\nu}f)|_\Sigma=0.
\end{equation}
Then the development of the AFG equation is almost Einstein with a covariant tractor given by
\begin{equation}
X_0=f,\ X_\mu=\partial_\mu f,\ X_\infty=-\frac{1}{d}(\nabla^\mu\nabla_\mu f+P^\mu_\mu f)
\end{equation}
where $f$ is the  solution of the scalar GJMS equation of weight $1$:
\begin{equation}
G_1f=0
\end{equation}
with the given initial data. The solution is computed in an arbitrary Ricci flat extension.
\end{prop}

\begin{proof}
We need to prove \eqref{eq:DD}. As $\mboxplus$ is recursive till order $d/2-2$ it is enough to show (due to Lemma \ref{lm:bijective}) that
\begin{equation}
D^{d-3} [\mnabla_I\mnabla_J\mf]|_{t=1,\Sigma}^{[0]}=0
\end{equation}
A symbol $D^n$ denote jets in directions of $M$.
We compute using $[\mnabla_I\mnabla^I\mf]^{[0]}=0$
\begin{equation}
k:=[\mf]^{[1]}|_{t=1}=-\frac{1}{d}(\nabla^\mu\nabla_\mu f+P^\mu_\mu f)
\end{equation}
Let us now define $\mX_I=\partial_I \mf$, $X_I=[\mX_I]^{[0]}|_{t=1}$.
We notice that 
\begin{equation}
X_0=f,\ X_\mu=\partial_{\mu}f,\ X_\infty=k.
\end{equation}
The condition \eqref{eq:tractor-cond-D} gives  by \eqref{eq:tractor-connection} 
\begin{equation}\label{eq:Dd-2}
D^{d-2}\nabla^{\mathcal T}_{\iI} X_J|_\Sigma=0.
\end{equation}
The nontrivial condition for $\nabla^{\mathcal T}_{\iI} X_\infty$ is shown by divergence of \eqref{eq:tractor-cond-D} (see \cite{Curry2015}). For this reason, we only get the condition for $D^{d-2}$ jets. The equation \eqref{eq:Dd-2} means
\begin{equation}\label{eq:tractor-1}
D^{d-2}(\mnabla_I\mnabla_J \mf)|_{t=1,\Sigma}=D^{d-2}(\mT_I\mF_J+\mr\mD_{IJ})|_{t=1,\Sigma},
\end{equation}
for some $\mF_I$ and $\mD_{IJ}$ such that $\Lie_{\mT}\mF=-2\mF$ and $\Lie_{\mT}\mD=-\mD$ in order that $\Lie_{\mT}\mnabla\mX=\mnabla\mX$. Moreover, from $\mT^I\mnabla_I\mnabla_J \mf=0$ we have
\begin{equation}
D^{d-2}(2\mr\mF_J+\mr \mT^I\mD_{IJ})|_{t=1,\Sigma}=0\Longrightarrow D^{d-2}(\mT^I\mD_{IJ})|_{t=1,\Sigma}=-D^{d-2}(2\mF_J)|_{t=1,\Sigma}.
\end{equation}
Finally, due to $\RS_{IJ}=O(\rho^\infty)$ we derive
\begin{equation}
\mnabla_I\mnabla^I \mX_J=\mnabla_J (\mnabla_I\mnabla^I\mf)+O(\rho^\infty)=O(\rho^{d/2}).
\end{equation}
Substituting the form of $\mnabla_I\mX_J$ from \eqref{eq:tractor-1} we obtain
\begin{equation}
O(\rho^{d/2})=D^{d-3}[\mnabla^I(\mT_I\mF_J)+\mnabla^I(\mr\mD_{IJ})]|_{t=1,\Sigma}=
D^{d-3}((d-2)\mF_J+\mr \mnabla^I\mD_{IJ})|_{t=1,\Sigma},
\end{equation}
so $D^{d-3}[\mF_J]|_{t=1,\Sigma}^{[0]}=0$ which means that
\begin{equation}
D^{d-3}[\mnabla_I\mnabla_J \mf]|_{t=1,\Sigma}^{[0]}=0,
\end{equation}
and the initial condition \eqref{eq:DD} are satisfied.
\end{proof}

\begin{remark}
In the case $\mf^{[0]}|_{t=1}\not=0$ on the Cauchy surface, we can change conformally the metric such that it satisfies Einstein constraints on the initial surface. Therefore, we can  evolve Einstein equations. The result needs to agree with the metric evolved with AFG equation up to conformal rescaling and diffeomorphism. In this way we obtain propagation of the Einsteinian condition up to the conformal boundary (see \cite{Anderson}). Surface $\mf^{[0]}|_{t=1}=0$ is more delicate. Our method has the advantage that it allows to treat all cases simultaneously. For example, the initial Cauchy surface can cross the conformal boundary.
\end{remark}

\section{Summary}

The Fefferman-Graham obstruction tensor and GJMS operators are very special objects. One additional nice property is related to their behavior as evolution systems. Both GJMS as well as Fefferman-Graham tensor in the suitable gauge are not strongly hyperbolic, but still they enjoy well-posed Cauchy problem.  In addition, the property of being almost Einstein propagates from the initial surface. We proved it in the smooth category, but with an arbitrary Cauchy surface. Namely, the Cauchy surface can cross or partially coincide with the conformal boundary of the spacetime. This allows to use AFG equation for proving the stability of asymptotically simple solutions (as was done in \cite{Anderson, Anderson-Chrusciel}), which was the initial motivation for studying AFG equations. We should notice that this is not the most effective proof of stability as the metric needs to be of high regularity. However, it provides some advantages: it is a Lagrangean theory, which allows to apply various techniques like Noether charge definition, Hamiltonian formulations on the level of conformally compactified spacetime. The meaning of such defined charges for Einsteinian solutions is still unclear. The relation to GR charges should be investigated in future.

\section*{Acknowledgement}

The author would like to thank Piotr Chru{\'s}ciel for very useful discussions, comments, and help during the development of this work. This work was supported by Project OPUS 2017/27/B/ST2/02806 of Polish National Science Centre. Data sharing not applicable to this article as no datasets were generated or analysed during the current study.

\appendix

\section{Bianchi identities}\label{sec:Bianchi}

For convenience of the reader, we provide here a proof of the Bianchi identities in the ambient space. Our first goal is to show that

\begin{prop}\label{prop:Bianchi}
The following holds
\begin{align}
&\tnabla^{\iI} \left(\RS_{\iI\jI}-\frac{1}{2}\g_{\iI\jI}\RS\right)+\rho\partial_{\iI}\RS_{\infty\infty}+(d-2-2\rho\partial_\infty)\RS_{\iI\infty}-\rho\g^{\kI\lI}\g_{\kI\lI}'\RS_{\iI\infty}=0,\label{eq:B-1}\\
&\tnabla^{\iI}\RS_{\iI\infty}+(d-2-\rho\partial_\infty)\RS_{\infty\infty}-\rho \g^{\iI\jI}\g_{\iI\jI}'\RS_{\infty\infty}-\frac{1}{2}\g_{\iI\jI}'\RS^{\iI\jI}-\frac{1}{2}\partial_\infty \RS=0,\label{eq:B-2}
\end{align}
where $\RS=\RS^{\iI}_{\iI}$.
\end{prop}

The proof will be divided in a series of lemmas.
Let us denote by $\mg_{IJ}$, $\mnabla_{I}$, $\mRS_{IJ}$ metric covariant derivative and Ricci tensor respectively in the ambient space $\mM$ . Indices are $I=0,\infty$ or $\iI$ in the case of index on $M$.

\begin{lm}\label{lm:div1}
Let $\tF_{\mu}$ be a form and $\tF_\infty$ a function on $\tilde{M}$. Define a form $\mF_I$ on $\mM$ by
\begin{equation}
\mF_{\iI}=\tF_{\iI},\quad \mF_\infty=\tF_\infty,\quad \tF_0=0.
\end{equation}
Then
\begin{equation}
\mnabla^I\mF_I=t^{-2}\left(\tnabla^{\iI} \tF_{\iI}+(d-2-2\rho\partial_\infty)\tF_\infty-\rho\g^{\iI\jI}\g_{\iI\jI}' \tF_\infty\right).
\end{equation}
\end{lm}

\begin{proof}
Let us notice the identity
\begin{equation}
{\mg}^{IJ}\mnabla_I \mF_J=\frac{1}{\sqrt{\mg}}\partial_{I}(\sqrt{\mg}{\mg}^{IJ}\mF_J).
\end{equation}
Now $\mg^{0\infty}=\mg^{\infty 0}=t^{-1}$, $\mg^{\infty\infty}=-2\rho t^{-2}$ and $\mg^{\iI\jI}=\g^{\iI\jI}$ the rest of the components vanishes. Moreover, $\sqrt{\mg}=t^{d+1}\sqrt{\g}$. Thus, remembering that $\mF_0=0$ it follows that
\begin{align}
&{\mg}^{IJ}\mnabla_I \mF_J=t^{-2}\frac{1}{\sqrt{\g}}\partial_{\iI}(\sqrt{\g}{\g}^{\iI\jI}\tF_{\jI})+
\frac{1}{t^{d+1}\sqrt{\g}}\partial_0(t^{d+1}\sqrt{\g}t^{-1}\tF_\infty)+
\frac{1}{\sqrt{\g}}\partial_\infty(\sqrt{\g}(-2\rho t^{-2})\tF_\infty)=\nonumber\\
&=t^{-2}\left(\tnabla^{\iI} \tF_{\iI}+(d-2-2\rho\partial_\infty)\tF_\infty-\rho\g^{\iI\jI}\g_{\iI\jI}' \tF_\infty\right),
\end{align}
where we used $\frac{1}{\sqrt{\g}}\partial_{\iI}(\sqrt{\g}{\g}^{\iI\jI}\tF_{\jI})=\tnabla^{\iI} \tF_{\iI}$.
\end{proof}

\begin{lm}\label{lm:div2}
Let $\mD_{IJ}$ be a symetric tensor and $\mX$ a vector field, then
\begin{equation}
(\mnabla^I \mD_{IJ})\mX^J=\mnabla^I (\mD_{IJ}\mX^J)-\frac{1}{2}\mD^{IJ}{\mathcal L}_{\mX}\mg_{IJ},
\end{equation}
where ${\mathcal L}_{\mX}$ is a Lie derivative
\end{lm}

\begin{proof}
Follows from ${\mathcal L}_{\mX}\mg_{IJ}=\mnabla_I \mX_J+\mnabla_J \mX_I$ and symmetry of $\mD_{IJ}$.
\end{proof}

\begin{lm}
We have
\begin{align}
\mnabla^I \mRS_{I {\iI}}&=t^{-2}\left(\tnabla^{\jI} \RS_{\jI\iI}+(d-2-2\rho\partial_\infty)\RS_{ \iI\infty}-\rho\g^{\iI\jI}\g_{\iI\jI}' \RS_{\iI \infty}\right),\label{eq:S1}\\
\mnabla^I \mRS_{I\infty}&=t^{-2}\left(\tnabla^{\jI} \RS_{\jI\infty}+(d-2-2\rho\partial_\infty)\RS_{ \infty\infty}-\rho\g^{\iI\jI}\g_{\iI\jI}' \RS_{\infty \infty}-\frac{1}{2}\RS^{\jI\iI}\g_{\iI\jI}'-\RS_{\infty\infty}\right).\label{eq:S2}
\end{align}
\end{lm}

\begin{proof}
Let us choose first $\mX^{\iI}=X^{\iI}$, $\mX^0=0$ and $\mX^\infty=0$ for some $X^{\iI}$, vector field on $M$. The form
\begin{equation}
\mRS_{IJ}\mX^J
\end{equation}
satisfies assumptions of Lemma \ref{lm:div1}, thus
\begin{align}
&\mnabla^I (\mRS_{I J}\mX^J)=t^{-2}\left(\tnabla^{\kI} (\RS_{\kI\jI}X^{\jI})+(d-2-2\rho\partial_\infty)\RS_{\infty \iI}X^{\iI}-\rho\g^{\iI\jI}\g_{\iI\jI}' \RS_{\infty \kI}X^{\kI}\right)=\nonumber\\
&=t^{-2}\left(X^{\iI}\tnabla^{\kI} \RS_{\kI\iI}+\RS_{\kI\jI}(\tnabla^{\kI}X^{\jI})+(d-2-2\rho\partial_\infty)\RS_{\infty \iI}X^{\iI}-\rho\g^{\iI\jI}\g_{\iI\jI}' \RS_{\infty \kI}X^{\kI}\right).
\end{align}
Moreover, we have
\begin{equation}
( {\mathcal L}_{\mX} \mg_{IJ})dx^Idx^J=t^2( {\mathcal L}_{X} \g_{\kI\jI})dx^{\kI}dx^{\jI}=t^2 2(\tnabla_{\kI} X_{\jI})dx^{\kI}dx^{\jI},
\end{equation}
By Lemma \ref{lm:div2} we derive
\begin{align}
&\mX^J\mnabla^I \mRS_{I J}=t^{-2}X^{\iI}\left(\tnabla^{\kI} \RS_{\kI\iI}+(d-2-2\rho\partial_\infty)\RS_{\infty \iI}-\rho\g^{\kI\jI}\g_{\kI\jI}' \RS_{\infty \iI}\right),
\end{align}
which shows \eqref{eq:S1}.
Similarly choosing $\mX^I=\delta^I_\infty$ we have 
\begin{equation}
( {\mathcal L}_{\mX} \mg_{IJ})dx^Idx^J=
2dt^2+t^2\g_{\iI\jI}'dx^{\iI}dx^{\jI},
\end{equation}
where $\g_{\iI\jI}'$ is the derivative in $\rho$.
Finally, we obtain
\begin{equation}
\mnabla^I (\mRS_{I J}\mX^J)=t^{-2}\left(\tnabla^{\kI} \RS_{\kI\infty}+(d-2-2\rho\partial_\infty)\RS_{\infty \infty}-\rho\g^{\iI\jI}\g_{\iI\jI}' \RS_{\infty \infty}\right).
\end{equation}
Thus we conclude
\begin{align}
\mnabla^I \mRS_{I\infty}&=t^{-2}\Big(\tnabla^{\jI} \RS_{\jI\infty}+(d-2-2\rho\partial_\infty)\RS_{ \infty\infty}-
\rho\g^{\iI\jI}\g_{\iI\jI}' \RS_{\infty \infty}-\frac{1}{2}\RS^{\iI\jI}\g_{\iI\jI}'-\RS_{\infty\infty}\Big),
\end{align}
which shows the result.
\end{proof}

\begin{proof}[Proof of Proposition \ref{prop:Bianchi}]
Let us now notice that
\begin{equation}
\mg^{IJ}\mRS_{IJ}=t^{-2}(\g^{\iI\jI}\RS_{\iI\jI}-2\rho\RS_{\infty\infty}).
\end{equation}
Now we use identity
\begin{equation}
\mnabla^I \mRS_{IJ}-\frac{1}{2}\partial_{J} (\mg^{KL}\mRS_{KL})=0
\end{equation}
to get
\begin{align}
&\tnabla^{\iI} \left(\RS_{\iI\jI}-\frac{1}{2}\g_{\iI\jI}\RS\right)+\rho\partial_{\iI}\RS_{\infty\infty}+(d-2-2\rho\partial_\infty)\RS_{\iI\infty}-\rho\g^{\kI\lI}\g_{\kI\lI}'\RS_{\iI\infty}=0\\
&\tnabla^{\iI}\RS_{\iI\infty}+(d-2-\rho\partial_\infty)\RS_{\infty\infty}-\rho \g^{\iI\jI}\g_{\iI\jI}'\RS_{\infty\infty}-\frac{1}{2}\g_{\iI\jI}'\RS^{\iI\jI}-\frac{1}{2}\partial_\infty \RS=0,
\end{align}
which is the result.
\end{proof}

Let us remind
\begin{align}
&\tB^1_{\iI}=-\frac{1}{2}\tnabla^{\kI}\tnabla_{\kI} \tG_{\iI}-\frac{1}{2}\RR_{\iI}^{\jI}\tG_{\jI}-\left(\frac{d}{2}-1-\rho\partial_\infty\right)\partial_\infty \tG_{\iI}+\frac{1}{2}\g^{\kI\lI}\g'_{\kI\lI}\rho\partial_\infty \tG_{\iI}+\frac{1}{2}\rho\g^{\kI\lI}\g'_{\kI\lI}\partial_{\iI}\tS,\\
&\tB^2=-\frac{1}{2}\tnabla^{\iI}\partial_{\iI}\tS
-\left(\frac{d}{2}-2-\rho\partial_\infty\right)\partial_\infty\tS+\g^{\iI\jI}\g_{\iI\jI}'\rho\partial_\infty\tS
+\frac{1}{2}\tQ^{\iI}\tG_{\iI}+\frac{1}{2}\g_{\iI\jI}'\tnabla^{\iI} \tG^{\jI}+\frac{1}{2}\g_{\iI\jI}'\g^{\iI\jI}\tS,
\end{align}
where $\tQ^{\iI}=\partial_\infty (\g^{\jI\kI}\tilde{\Gamma}^{\iI}_{\jI\kI})$.
We will prove the following important properties of these objects.

\begin{prop}\label{lm:B-follows}
Suppose that $\tE_{\iI\jI}=O(\rho^{d/2})$ then $\tB^1_{\iI}=O(\rho^{d/2})$ and $\tB^2=O(\rho^{d/2-1})$.
\end{prop}

\begin{proof}
We will use the Bianchi identity \eqref{eq:B-1} and \eqref{eq:B-2}.
Let us now also notice
\begin{equation}
\tnabla^{\iI} \left(\tnabla_{\iI}\tG_{\jI}+\tnabla_{\jI}\tG_{\iI}-\g_{\iI\jI}\tnabla^{\kI}\tG_{\kI}\right)=\tnabla^{\kI}\tnabla_{\kI} \tG_{\jI}+\RR_{\jI}^{\iI}\tG_{\iI}.
\end{equation}
Then we compute
\begin{align}
&\tnabla^{\iI} \left(\RS_{\iI\jI}-\frac{1}{2}\g_{\iI\jI}\RS\right)
=\tnabla^{\iI} \left(\tE_{\iI\jI}-\frac{1}{2}\g_{\iI\jI}\tE\right)+\frac{1}{2}(\tnabla^{\kI}\tnabla_{\kI} \tG_{\jI}+\RR_{\jI}^{\iI}\tG_{\iI})-\left(\frac{d}{2}-1\right)\partial_{\iI}\tS.\label{eq:Bianchi}
\end{align}
Moreover, differentiating \eqref{eq:tS} and \eqref{eq:tG} with respect to $\rho$ we get
\begin{equation}\label{eq:tS-tG}
\RS_{\infty\infty}=\partial_\infty\tS,\quad \RS_{\iI\infty}=\frac{1}{2}(\partial_\infty \tG_{\iI}+\partial_{\iI}\tS).
\end{equation}
Remembering that $\tE=O(\rho^{d/2})$ and
combining \eqref{eq:tS-tG} with \eqref{eq:B-1} and \eqref{eq:Bianchi} we obtain
\begin{equation}
\frac{1}{2}\tnabla^{\kI}\tnabla_{\kI} \tG_{\iI}+\frac{1}{2}\RR_{\iI}^{\jI}\tG_{\jI}+\left(\frac{d}{2}-1-\rho\partial_\infty\right)\partial_\infty \tG_{\iI}-\frac{1}{2}\g^{\kI\lI}\g'_{\kI\lI}\rho\partial_\infty \tG_{\iI}-\frac{1}{2}\rho\g^{\kI\lI}\g'_{\kI\lI}\partial_{\iI}\tS=O(\rho^{d/2}).
\end{equation}
Similar computation with
\begin{align}
&\RS=\tE+\tnabla^{\iI}\tG_{\iI}+d\tS=\tnabla^{\iI}\tG_{\iI}+d\tS+O(\rho^{d/2}),\\
&\g_{\iI\jI}'\RS^{\iI\jI}=\g_{\iI\jI}'\tE^{\iI\jI}+
\g_{\iI\jI}'\tnabla^{\iI} \tG^{\jI}+\g_{\iI\jI}'\g^{\iI\jI}\tS=
\g_{\iI\jI}'\tnabla^{\iI} \tG^{\jI}+\g_{\iI\jI}'\g^{\iI\jI}\tS
+O(\rho^{d/2}),
\end{align}
reveals after inserting it into \eqref{eq:B-2} 
\begin{align}
&\frac{1}{2}\tnabla^{\iI}(\partial_{\iI}\tS+\partial_\infty\tG_{\iI})+(d-2-\rho\partial_\infty)\partial_\infty\tS-\g^{\iI\jI}\g_{\iI\jI}'\rho\partial_\infty\tS+\nonumber\\
&-\frac{1}{2}\partial_\infty \tnabla^{\iI}\tG_{\iI}-\frac{d}{2}\partial_\infty\tS -\frac{1}{2}\g_{\iI\jI}'\tnabla^{\iI} \tG^{\jI}-\frac{1}{2}\g_{\iI\jI}'\g^{\iI\jI}\tS=O(\rho^{d/2-1}).
\end{align}
Taking into account that 
\begin{equation}
\partial_\infty \tnabla^{\iI} \tG_{\iI}-\tnabla^{\iI} \partial_\infty \tG_{\iI}=\tQ^{\iI}\tG_{\iI},\quad \tQ^{\iI}=\partial_\infty \g^{\jI\kI}\tilde{\Gamma}^{\iI}_{\jI\kI},
\end{equation}
we conclude
\begin{equation}
\frac{1}{2}\tnabla^{\iI}\partial_{\iI}\tS
+\left(\frac{d}{2}-2-\rho\partial_\infty\right)\partial_\infty\tS-\g^{\iI\jI}\g_{\iI\jI}'\rho\partial_\infty\tS
-\frac{1}{2}\tQ^{\iI}\tG_{\iI}
-\frac{1}{2}\g_{\iI\jI}'\tnabla^{\iI} \tG^{\jI}-\frac{1}{2}\g_{\iI\jI}'\g^{\iI\jI}\tS=O(\rho^{d/2-1}),
\end{equation}
and the proposition is proven.
\end{proof}


\begin{thebibliography}{10}

\bibitem{Friedrich1983}
H.~Friedrich, ``Cauchy problems for the conformal vacuum field equations in
  general relativity,'' {\em Communications in Mathematical Physics}, vol.~91,
  no.~4, pp.~445--472, 1983.

\bibitem{Anderson}
M.~T. Anderson, ``Existence and {Stability} of {Even}-dimensional
  {Asymptotically} de {Sitter} {Spaces},'' {\em Annales Henri Poincaré},
  vol.~6, pp.~801--820, Oct. 2005.

\bibitem{Anderson-Chrusciel}
M.~T. Anderson and P.~T. Chru{\'s}ciel, ``Asymptotically simple solutions of
  the vacuum {E}instein equations in even dimensions,'' {\em Comm. Math. Phys},
  vol.~260, pp.~557--577, 2005.

\bibitem{Fefferman-Graham-1}
C.~Fefferman and C.~R. Graham, ``Conformal invariants,'' in {\em \'Elie Cartan
  et les math\'ematiques d'aujourd'hui - Lyon, 25-29 juin 1984}, no.~S131 in
  Ast\'erisque, Soci\'et\'e math\'ematique de France, 1985.

\bibitem{Fefferman-Graham}
C.~Fefferman and C.~R. Graham, {\em The {Ambient} {Metric} ({AM}-178)}.
\newblock Princeton University Press, 2012.

\bibitem{Branson}
T.~P. Branson, ``Sharp inequalities, the functional determinant, and the
  complementary series,'' {\em Transactions of the American Mathematical
  Society}, vol.~347, no.~10, pp.~3671--3742, 1995.

\bibitem{Guenther}
P.~G{\"u}nther, ``Über das {Cauchysche} {Problem} für die {Bachschen}
  {Feldgleichungen},'' {\em Mathematische Nachrichten}, vol.~69, no.~1,
  pp.~39--56, 1975.

\bibitem{Choquet-Bruhat}
Y.~Choquet-Bruhat, {\em General {Relativity} and the {Einstein} {Equations}}.
\newblock Oxford University Press, Oxford, Dec. 2008.

\bibitem{GJMS1}
C.~R. Graham, R.~Jenne, L.~J. Mason, and G.~A.~J. Sparling, ``Conformally
  {Invariant} {Powers} of the {Laplacian}, {I}: {Existence},'' {\em Journal of
  the London Mathematical Society}, vol.~s2-46, no.~3, pp.~557--565, 1992.

\bibitem{BransonQ}
T.~Branson, ``The functional determinant,'' {\em Global Analysis Research
  Center Lecture Notes Series, Seoul National University}, vol.~4, 1993.

\bibitem{Gover2004}
A.~R. Gover, ``Almost conformally {Einstein} manifolds and obstructions,'' {\em
  arXiv:math/0412393}, Dec. 2004.
\newblock arXiv: math/0412393 version: 1.

\bibitem{Graham2011}
C.~R. Graham and T.~Willse, ``Parallel tractor extension and ambient metrics of
  holonomy split $g_2$,'' {\em Journal of Differential Geometry}, vol.~92,
  pp.~463--506, Nov. 2012.

\bibitem{Curry2015}
S.~N. Curry and A.~R. Gover, ``An {Introduction} to {Conformal} {Geometry} and
  {Tractor} {Calculus}, with a view to {Applications} in {General}
  {Relativity},'' in {\em {A}symptotic {A}nalysis in {G}eneral {R}elativity}
  (D.~Häfner, J.-P. Nicolas, and T.~Daudé, eds.), London {Mathematical}
  {Society} {Lecture} {Note} {Series}, pp.~86--170, Cambridge: Cambridge
  University Press, 2018.

\bibitem{Ringstroem2009}
H.~Ringström, {\em The {Cauchy} {Problem} in {General} {Relativity}}.
\newblock June 2009.

\bibitem{taylor}
M.~E. Taylor, ``Partial {Differential} {Equations}. {III},'' {\em Applied
  Mathematical Sciences}, vol.~117, 1996.

\bibitem{Cap2002}
A.~{\v C}ap and A.~R. Gover, ``Standard {Tractors} and the {Conformal}
  {Ambient} {Metric} {Construction},'' {\em Annals of Global Analysis and
  Geometry}, vol.~24, pp.~231--259, Oct. 2003.

\bibitem{Chang2008}
S.-Y.~A. Chang, M.~Eastwood, B.~Ørsted, and P.~C. Yang, ``What is
  {Q}-{Curvature}?,'' {\em Acta Applicandae Mathematicae}, vol.~102,
  pp.~119--125, July 2008.

\bibitem{Baum2010}
H.~Baum and A.~Juhl, {\em Conformal {Differential} {Geometry}: {Q}-{Curvature}
  and {Conformal} {Holonomy}}.
\newblock Oberwolfach {Seminars}, Birkhäuser Basel, 2010.

\bibitem{Fefferman2003}
C.~Fefferman and K.~Hirachi, ``Ambient metric construction of {Q}-curvature in
  conformal and {CR} geometries,'' {\em Mathematical Research Letters},
  vol.~10, no.~6, pp.~819--831, 2003.
\newblock arXiv: math/0303184 version: 2.

\bibitem{Gover2008}
A.~R. Gover, ``Almost {Einstein} and {Poincaré}–{Einstein} manifolds in
  {Riemannian} signature,'' {\em Journal of Geometry and Physics}, vol.~60,
  pp.~182--204, Feb. 2010.

\bibitem{Gover2004a}
A.~R. Gover and P.~Nurowski, ``Obstructions to conformally {Einstein} metrics
  in n dimensions,'' {\em Journal of Geometry and Physics}, vol.~56,
  pp.~450--484, Mar. 2006.

\end{thebibliography}
\bibliographystyle{ieeetr}

\end{document}